%% file: main.tex
\documentclass[twoside]{article}

\usepackage[accepted]{aistats2026}
\input{sty}

\begin{document}

%

%
\runningauthor{Sirui Lin, Zijun Gao, Jos\'{e} Blanchet, Peter Glynn}

\twocolumn[

\aistatstitle{Causal Partial Identification via Conditional Optimal Transport}

\aistatsauthor{
Sirui Lin$^{1, *}$ \And
Zijun Gao$^{2, *}$ \And
Jos\'{e} Blanchet$^{1}$ \And
Peter Glynn$^{1}$
}

\aistatsaddress{
$^{1}$Department of Management Science and Engineering, Stanford University, USA \\
$^{2}$Marshall School of Business, University of Southern California, USA
}
]

\begin{abstract}
  We study the estimation of causal estimand involving the joint distribution of treatment and control outcomes for a single unit. In typical causal inference settings, it is impossible to observe both outcomes simultaneously, which places our estimation within the domain of partial identification (PI). Pre-treatment covariates can substantially reduce estimation uncertainty by shrinking the partially identified set. Recent work has shown that covariate-assisted PI sets can be characterized through conditional optimal transport (COT) problems. However, finite-sample estimation of COT poses significant challenges, primarily because the COT functional is discontinuous under the weak topology, rendering the direct plug-in estimator inconsistent. To address this issue, existing literature relies on relaxations or indirect methods involving the estimation of non-parametric nuisance statistics. In this work, we demonstrate the continuity of the COT functional under a stronger topology induced by the adapted Wasserstein distance. Leveraging this result, we propose a direct, consistent, non-parametric estimator for COT value that avoids nuisance parameter estimation. We derive the convergence rate for our estimator and validate its effectiveness through comprehensive simulations, demonstrating its improved performance compared to existing approaches.
\end{abstract}

\section{Introduction}
The potential outcome model \cite{rubin1974estimating, imbens2015causal} is prevalent in causal inference, where each unit is associated with two potential outcomes: one under treatment and one under control. 
Since the two potential outcomes are never observed simultaneously, their joint distribution is not identified, making the causal estimands that depend on the joint distribution only partially identifiable.
Optimal transport~\cite{villani2009optimal}, which minimizes an expected cost over joint distributions respecting the marginals, has been used to recover partial identification (PI) sets of causal estimands~\cite{gao2024bridging}.
When pre-treatment covariates are available, they can often be used to reduce uncertainty about the joint distribution of potential outcomes, and produce smaller, more informative PI sets.
Conditional optimal transport (COT) has been used to characterize covariate-assisted PI sets~\cite{ji2023model, fan2025partial, lin2025tightening, balakrishnan2023conservative} by finding the optimal joint distribution preserving the outcome-covariate distributions.

Despite the potential to yield more meaningful covariate-assisted PI sets, COT poses more substantial statistical challenges compared to OT, especially with continuous covariates.
First, it is rare to observe a shared continuous covariate value in both treatment and control groups, leaving the conditional marginal distribution ill-defined at that value for at least one group. Specifically, when conditioning on an exact covariate value, there will typically be only a single unit in the sample that attains that value, with matched or repeated observations being unlikely.
Second, even if each covariate value is associated with both a treated and a control unit (such as in twin studies), there is typically only one observation per covariate value in each group.
The plug-in estimator of COT value, which uses point masses to estimate the conditional marginal distributions, may not converge to the true COT value as the sample size goes to infinity (\cite{chemseddine2024conditional} and Example~\ref{exp:counter}).
This issue suggests that the COT functional may not be continuous under the weak topology, rendering the COT value estimation more challenging than standard OT problems.

In this paper, we demonstrate the continuity of the COT functional under a stronger topology induced by the adapted Wasserstein distance, which requires uniform convergence of the outcome's conditional distribution at each covariate value. 
Leveraging the continuity result, we propose a direct, consistent, non-parametric estimator for the COT value. 
Our estimator is constructed using an adapted empirical distribution, which discretizes the covariate space by assigning observed covariates to a finite number of cells. 
The convergence of the adapted empirical distribution under the adapted Wasserstein distance ensures the consistency of our estimator, and we further establish its finite-sample error bounds.
Compared to existing estimators for COT value that rely on relaxations of the original COT formulation~\cite{manupriya2024consistent, lin2025tightening}, or utilize indirect approaches that require estimating nuisance functions~\cite{ji2023model}, or resort to heuristic procedures without statistical guarantees~\cite{tabak2021data, bunne2022supervised, hosseini2023conditional}, our proposed COT value estimator is direct and statistically consistent. 

\noindent\textbf{Contributions}. 
\begin{enumerate}
    \item We establish the COT formulation for the covariate-assisted PI sets of a class of causal estimands (Theorem~\ref{thm:optimal_causalbounds}). We prove the continuity of the COT functional under the topology induced by the adapted Wasserstein distance (Theorem~\ref{prop:con_adapted_topology}).

    \item Building on (a), under Bernoulli design, we propose a direct nonparametric estimator for the covariate-assisted PI set, prove its consistency (Theorem~\ref{thm:consist}), and further establish the finite-sample convergence rate (Theorem~\ref{thm:finite_sample}) together with a robustness guarantee under covariate's distributional shift (Theorem~\ref{rmk:robust}).
    We then extend the result to unconfounded designs with covariate-dependent treatment assignment (Theorem~\ref{thm:finite_sampleII}).
    In addition to its relevance for causal inference, our work advances the literature on COT by providing a consistent non-parametric estimator for the COT value with provable convergence rate.

    \item We validate the estimator's effectiveness through comprehensive empirical studies (Section~\ref{sec:simulation}). 
\end{enumerate}
\noindent \textbf{Organization}. The rest of the paper is structured as follows. 
Section~\ref{sec:problem_form} introduces the problem setting. 
Section~\ref{sec:connection} provides basic properties of COT values, which serve as building blocks for our proposal.
Section~\ref{sec:main_result} presents the construction of our COT value estimator and the main theoretical results on its finite-sample performance. In Section~\ref{sec:simulation}, we provide simulation studies comparing our method with existing approaches. Technical proofs are deferred to the appendix.

\noindent \textbf{Notation}. We use $Z(w)$ to denote the variable $Z$ observed within the control ($w=0$) and treatment ($w=1$) groups.
We denote $[n] = \{1,...,n\}, n\wedge m = \min(n,m), n \vee m = \max(n,m)$. We denote by $\calP(\Omega)$ the set of probability measures on $\Omega$. We use $\mu_{Y,Z}$ to denote the probability distribution (and its associated measure on $\calY \times \calZ$) of $(Y, Z)$ under $\mu$, with $\mu_{Z}$ denoting the marginal distribution of $Z$, and $\mu^{z}_{Y}$ the conditional distribution of $Y$ given $Z=z$. By $\mu = \mu_Z \otimes \mu_Y^{Z}$, we mean that  for any measurable function $g$, $\int g \diff \mu = \int_{\calZ} \int_{\calY} g(z,y) \diff \mu_Y^{z}(y) \diff \mu_Z(z)$. We denote the set of the joint couplings $\pi$ with marginals $\mu, \nu \in \calP(\calX)$ ($\calX$ may be $\calY, \calZ$ or $\calY \times \calZ$) by
\begin{equation}\label{eq:coupling_set}
    \Pi(\mu, \nu) \Let \left\{\pi \in \calP(\calX^2): \pi_{X} = \mu,~\pi_{X'} = \nu \right\}.
\end{equation}

\section{Problem Formulation}\label{sec:problem_form}
\subsection{Potential Outcome Model}\label{sec:potential.outcome.model}
We consider the potential outcome model with a binary treatment \cite{rubin1974estimating}. 
Suppose there are $N$ units.
Each unit is associated with two potential outcomes, \( Y_i(0) \) and \( Y_i(1) \in \calY \subseteq \mathbb{R}^{d_Y} \), $i \in [N]$, $d_Y \ge 1$,  $\calY$ a convex closed set. Here \( Y_i(0) \) is the outcome under control and \( Y_i(1) \) is the outcome under treatment\footnote{We allow the potential outcomes to be multivariate to accommodate trials with multiple endpoints, experiments with a primary and secondary outcome, or policy evaluations involving effectiveness and cost.}.
Each unit receives a binary treatment \( W_i \in \{0,1\} \), where \( W_i = 1 \) indicates treatment and \( W_i = 0 \) indicates control. 
For unit \( i \), only the outcome  \( Y_i(W_i) \) is observed.

Often, each unit is also associated with a covariate vector \( Z_i \in \calZ \subseteq \mathbb{R}^{d_Z} \), \( d_Z \ge 1 \), $\calZ $ a convex closed set. Covariates may often include continuous components, such as age or lab measurements. 
We denote the probability of receiving treatment conditional on the covariate value, typically referred to as propensity score, by \( e(z) := \mathbb{P}(W_i = 1 \mid Z_i = z) \).

We impose the following standard assumptions \cite{imbens2015causal} for the potential outcome model.
\begin{assumption}[Stable Unit Treatment Value Assumption (SUTVA)]\label{assu:SUTVA}
    Unit $i$'s potential outcomes only depend on $W_i$ but not $W_j$, $j \neq i$.
\end{assumption}

\begin{assumption}[Unconfoundedness]\label{assu:unconfoundedness}
    Treatment assignment is independent of potential outcomes given covariates, that is \( Y_i(w) \independent W_i \mid Z_i, \ w=0,1 \).
\end{assumption}

\begin{assumption}[Overlap]\label{assu:overlap} 
    There exists \( \delta > 0 \) such that for any possible covariate value $z$, the propensity score satisfies \( \delta < e(z) < 1 - \delta \).
\end{assumption}

To enable the asymptotic analysis, 
we adopt the super-population model~\cite{imbens2015causal} justified by Assp.~\ref{assu:SUTVA}, \ref{assu:unconfoundedness},
\begin{align*}
    (Y_i(0), Y_i(1), Z_i) 
    \stackrel{\text{i.i.d.}}{\sim} \mu \Let \mu_{Y(0), Y(1), Z},
\end{align*}
where \( \mu \) is the unknown joint distribution. 
The following result shows that, under the super-population model with unconfounded treatment assignment, the outcome-covariate distributions \( (Y_i(1), Z_i) \) and \( (Y_i(0), Z_i) \) are identifiable from the observed data.  
The outcome-covariate distributions will be later used in the COT formulation of the PI sets.
\begin{proposition}[Identifiable marginals]\label{prop:ID.marginal}
    Under~Assp. \ref{assu:SUTVA} - \ref{assu:overlap}, for $w=0,1$,
    \begin{align}\label{eq:ID.marginal}
        &\PP\left(Y_i(w) = y, Z_i = z\right) \notag\\ 
        =~& \tilde e_w(z)\  \PP\left(Y_i(W_i) = y, Z_i = z \mid W_i = w \right), 
    \end{align}
    where $\tilde e_w(z) = \frac{(1-w) - \bar{e}}{(1-w) - e(z)}$,  $\bar{e} \Let \int e(z) \diff \mu_z(z)$.
\end{proposition}

\subsection{PI Sets of Causal Estimands}\label{sec:partial.identificaion}

In this paper, we focus on causal estimands of the form 
\begin{align*}
    V^* \Let \mathbb{E}_{\mu}[h(Y(0), Y(1))].
\end{align*} 
Here, \( h: \calY^2 \rightarrow \mathbb{R} \) is a pre-specified objective function.

A fundamental challenge in causal inference is that \( Y_i(0) \) and \( Y_i(1) \) cannot be observed simultaneously for the same unit, making the joint distribution \( \mu_{Y(0), Y(1)} \) inherently unidentifiable. As a result, for causal estimands depending on the joint distribution, their values are not identifiable from observed data\footnote{Non-identifiability refers to the fact that two different joint distributions can yield different values of the causal estimand while producing the same distribution over observable quantities.
}.
This issue is commonly known as \textit{partial identification}.

\begin{example}[Variance of treatment effects]\label{exam:variance}
   For \( h(y_0, y_1) = (y_0 - y_1)^2 \), the corresponding estimand represents the variance of the treatment effect, which can be used to construct a confidence interval for the average treatment effect \( \mathbb{E}[Y(1) - Y(0)] \) in the classic Neymanian framework \cite{neyman1923application}. This estimand is only partially identifiable and has been studied in~\cite{aronow2014sharp, balakrishnan2025conservative}. 
\end{example}


For partially identifiable causal estimands, point estimation is infeasible.
The object of interest becomes the partial identification (PI) set, that is, the set of values consistent with the marginal distributions of \( (Y_i(1), Z_i) \) and \( (Y_i(0), Z_i) \). 
Mathematically, PI set problems can be formulated under the OT framework (Section~\ref{sec:COT}). First, define the coupling set as
\begin{align*}
    \Pi_{\cc} \Let \left\{ \pi \in \mathcal{P}(\calY^2 \times \calZ): \pi_{Y(w), Z} = \mu_{Y(w), Z}\ \   w=0,1\right\}.
\end{align*}
Then, the PI set is given by
\begin{align}\label{eq:COT.PI}
    \left\{ \mathbb{E}_\pi[h(Y(0), Y(1))] : \pi \in \Pi_{\cc} \right\} = [V_{\cc}, \tilde V_{\cc}].
\end{align}
Here, we use the fact that the set \( \Pi_{\cc} \) is convex, the resulting PI set~\eqref{eq:COT.PI} is a convex subset of \( \mathbb{R} \), and thus an interval. 
To construct this interval, it suffices to compute its lower and upper bounds, starting with
\begin{align}\label{eq:COT.lower}
    V_{\cc} = \min_{\pi \in \Pi_{\cc}}  \mathbb{E}_\pi[h(Y(0), Y(1))],
\end{align}
and analogously the upper bound $\tilde V_{\cc}$ via a maximization over the same coupling set.

Note that, for certain objective functions $h$ (including Example~\ref{exam:variance}), Hoeffding--type results yield a closed-form optimizer for \eqref{eq:COT.lower} (see, e.g., \cite{tchen1980inequalities}). In contrast, we develop estimators for \eqref{eq:COT.lower} that apply to broad classes of smooth functions $h$.

\subsection{Conditional Optimal Transport}\label{sec:COT}
We first review optimal transport (OT), followed by a discussion of conditional optimal transport (COT).

OT has a rich literature \cite{villani2009optimal} with broad applications in machine learning and statistics. In particular, OT has recently emerged as a powerful tool in causal inference~\cite{black2020fliptest, charpentier2023optimal, de2024transport, torous2024optimal}. 
In an OT problem, the goal is to couple two distributions in a way that minimizes a specified cost function (e.g., $h$) while preserving their marginal distributions. The OT distance is defined as 
\[
\Wass_h(\mu, \nu) \Let \min_{\pi \in \Pi(\mu, \nu)} \EE_{\pi}[h(X, X')],
\]
where $\Pi(\mu, \nu)$ is defined in Eq.~\eqref{eq:coupling_set}. When $h(x,x') = \|x - x'\|_2$, the OT distance is called Wasserstein $1$-distance and denoted by $\Wass_1$. As discussed in Section~\ref{sec:partial.identificaion}, OT naturally connects to the PI set problem, where $\Wass_h(\mu_{Y(0)}, \mu_{Y(1)})$ can be interpreted as a lower bound of $V^*$.

Compared to OT, COT aligns conditional measures across values of the covariates, which
has seen various applications, such as conditional sampling~\cite{wang2023efficient, hosseini2023conditional, baptista2024representation}, domain adaptation~\cite{rakotomamonjy2022optimal}, and Bayesian flow matching~\cite{kerrigan2024dynamic, chemseddine2024conditional}.
In this paper, we focus on the natural connection between COT and the PI problem when covariates are available.
Particularly, recall that in the PI set problem with $Z_i$, it becomes essential to preserve the outcome-covariate distributions of \( (Y_i(0), Z_i) \) and \( (Y_i(1), Z_i) \), where \( Z_i \) is shared across both marginals. 
The corresponding covariate-assisted lower bound is defined in Eq.~\eqref{eq:COT.lower},
which, as shown below, is the optimal objective value of a COT problem.

\begin{proposition}[Recursive formulation]\label{prop:recursive_form}
    For a measurable function $h: \calY \times \calY \rightarrow \RR_{+}$, 
    we have
    \[
    V_{\cc} = \int \Wass_{h}\left(\mu_{Y(0)}^{z}, \mu_{Y(1)}^{z}\right) \diff \mu_Z(z).
    \]
\end{proposition}

\subsection{Adapted Wasserstein Distance and Its Topology}

We conclude this section by introducing the topology that is natural for the COT functional and that motivates our finite-sample estimator in Section~\ref{sec:main_result}.

As we will demonstrate in Section~\ref{sec:connection}, a key distinction between OT and COT as functionals lies in their topological properties. While the classical OT functional is continuous with respect to the outcome marginals under the weak topology (metrized by $\Wass_1$), the COT functional $V_{\cc}$ is generally discontinuous with respect to the joint outcome--covariate marginals under the weak topology. This discrepancy arises because COT is inherently a \emph{conditional} transport problem: its geometry is driven by transition kernels rather than full joint distributions.

A natural way to formalize this structure is to restrict the set of admissible couplings to those that respect the conditional structure of the marginals, i.e., to (bi-)causal or adapted couplings; see, e.g., \cite{pflug2012distance, lassalle2018causal, backhoff2020adapted}. From this perspective, the \textit{adapted} Wasserstein distance can be understood as the optimal transport distance obtained when couplings are constrained to match conditional distributions. This interpretation is particularly natural in causal inference settings, where one aims to compare conditional outcome laws given covariates rather than arbitrary joint couplings of $(Y,Z)$ and $(Y',Z')$.

We therefore equip $\calP(\calY \times \calZ)$ with the topology induced by the adapted Wasserstein distance.

\begin{definition}[Adapted Wasserstein distance]\label{defi:adapted_wassI}
    For $\mu, \nu \in \calP(\calY \times \calZ)$, the adapted Wasserstein distance between $\mu$ and $\nu$ is defined as  
    \begin{align*}
        & \Wass_{\bc}(\mu, \nu) \\
        = & \min_{\pi \in \Pi(\mu_{Z}, \nu_Z)} \int \lnorm{z - z'}{2} + \Wass_{1}(\mu_Y^{z}, \nu_Y^{z'})~\diff \pi(z, z').
    \end{align*}
\end{definition}

Compared to the conventional Wasserstein distance, which measures proximity between full joint distributions of $(Y,Z)$ and $(Y',Z')$, the adapted Wasserstein distance quantifies discrepancies at the level of transition kernels. It therefore captures the transportation cost of conditional outcome laws while optimally aligning covariate marginals. As we show in Theorem~\ref{prop:con_adapted_topology}, under mild regularity conditions, $V_{\cc}$ becomes continuous with respect to this stronger, causally structured topology.

\section{Basic Properties of COT Value}\label{sec:connection}

The following assumption enables our analysis.
\begin{assumption}[Compact covariate domain]\label{a:compact_domain}
    Assume that $\calZ \subseteq [0,1]^{d_Z}$, and $\calY\subseteq \RR^{d_Y}$ is compact.
\end{assumption}
\begin{assumption}[Continuous objective]\label{a:cont_obj}
    Assume that $h: \calY\times\calY \rightarrow\RR$ is a continuous function. 
\end{assumption}
Here, we assume that the outcome $Y$ and the covariate $Z$ are supported on a compact set, which is often satisfied in practice—for example, when the outcome includes test scores, earnings, or health indicators and the covariate includes bounded physical measurements like age, height, or weight.

\subsection{Optimal Covariate-Assisted PI Sets}

In this section, we show that if we only have access to i.i.d.~samples drawn from $\mu_{Y(0), Z}$ and $\mu_{Y(1), Z}$, respectively and no additional side information, then the interval $[V_{\cc}, \tilde V_{\cc}]$ defined in \eqref{eq:COT.PI} is the \emph{optimal} PI set for $V^*$ that can be identified. Note that, in the remaining parts of this work, we will use $((Y_i(w), Z_i(w))$ to denote the outcome and covariate values collected from the sample $i$ within the control group ($w=0$) and treatment group ($w=1$). \textit{Note that in this work, $Z_i(w)$ only contains pre-treatment covariates, and the $w$ in $Z_i(w)$ is the group index of the observed pre-treatment covariates.}

\begin{theorem}[Optimal PI bounds]\label{thm:optimal_causalbounds}
    Under Assp.~\ref{a:compact_domain}-\ref{a:cont_obj},  suppose that the sample $((Y_i(w), Z_i(w)), i \in [n])$ i.i.d.~follows $\mu_{Y(w), Z}$ for $w=0,1$. Then, we have (i) $V_{\cc}, \tilde V_{\cc}$ are identifiable from the given sample, (ii) $[V_{\cc}, \tilde V_{\cc}]$ is the interval that exactly contains all possible values of $\EE_{\pi}[h(Y(0), Y(1))]$, where $\pi$ satisfies: $\pi_{Y(0), Z} = \mu_{Y(0), Z}, \pi_{Y(1), Z} = \mu_{Y(1), Z}$.
\end{theorem}

\subsection{Stronger Topology for COT Functional}

As indicated by Proposition~\ref{prop:recursive_form}, the estimation of $V_{\cc}$ requires estimating the conditional distributions $\mu_{Y(0)}^{z}, \mu_{Y(1)}^{z}$. Thus, a direct plug-in estimator does not yield consistent estimates of $V_{\cc}$, a fact also observed in \cite{chemseddine2024conditional}. To further illustrate this issue, we consider an idealized case where the $i$-th covariate in the treated group is identical to that in the control group, i.e., $Z_i(1) = Z_i(0)~\forall i$ (e.g.~a twin study), presenting an example adapted from \cite[Example 9]{chemseddine2024conditional}.
\begin{example}[Inconsistent estimator]\label{exp:counter}
    Let $n_0=n_1$ and $h=|y_0 - y_1|$. Suppose that under $\mu$, $Y(0)$, $Y(1)$, $Z$ are independent, and $\mu_{Y(0), Z} = \mu_{Y(1), Z}$, thus $V_{\cc} = 0$. Let $\hat \mu_{w} = \frac{1}{n_w}\sum_{i=1}^{n_w} \delta_{Y_i(w), Z_i(w)}, w=0,1$. Further, assume that (i) $Z_i(1) = Z_i(0)~\forall i$, and (ii) the distribution $\mu_Z$ has a density. Then,
    \begin{align*}
        & \lim_{n_0, n_1\rightarrow \infty} \min_{\pi \in \hat \Pi_{\cc}}\EE_{\pi}[h(Y(0), Y(1))] \\
        = &\EE_{\mu}[|Y(0) - Y(1)|] > 0 = V_{\cc}\qquad \text{a.s.},
    \end{align*}
    where $\hat \Pi_{\cc} = \left\{ \pi: \pi_{Y(0), Z} = \hat \mu_{0}, \pi_{Y(1), Z} = \hat \mu_{1}\right\}$. This is because we aim for a perfect match in $Z$ while almost surely the values of covariates in the treated group are distinct, i.e., $Z_i(1) \neq Z_j(1)$ for $ i \neq j$. Therefore, $\hat \Pi_{\cc}$ only contains a single coupling: $\frac{1}{n_0}\sum_{i} \delta_{Y_i(0), Y_i(1), Z_i(0)}$. Thus, the objective simplifies to $\frac{1}{n_0}\sum_i |Y_i(1) - Y_i(0)| $, which converges to the population quantity $\mathbb{E}_{\mu}\left[|Y(1) - Y(0)|\right]$.
\end{example}

From a topological perspective, if \( V_{\cc} \) is viewed as a functional of \( \mu_{Y(0), Z} \) and \( \mu_{Y(1), Z} \), Example~\ref{exp:counter} implies that \( V_{\cc} \) is discontinuous under the weak topology. However, under mild assumptions, \( V_{\cc} \) becomes continuous at $\mu$ under the stronger topology induced by the adapted Wasserstein distance (see Definition~\ref{defi:adapted_wassI}). 

\begin{theorem}[Continuity under $\Wass_{\bc}$]\label{prop:con_adapted_topology}
    Let $(\nu_n = \nu_{Y(0), Y(1), Z, n}, n \geq 1)$ be a sequence in $\calP(\calY^2 \times \calZ)$. Suppose $g_w(z) = \mu_{Y(w)}^{z}\ w=0,1$ are continuous under the weak topology and $h$ is a continuous function. Under Assp.~\ref{a:compact_domain}-\ref{a:cont_obj}, if $\lim_{n\rightarrow \infty}\Wass_{\bc}(\nu_{Y(w), Z,n}, \mu_{Y(w), Z}) = 0$ for $w=0,1$, then
    \begin{align*}
        &\lim_{n\rightarrow \infty}\int \Wass_{h}\left(\nu_{Y(0), n}^{z}, \nu_{Y(1), n}^{z}\right) \diff \nu_{Z, n}(z) \\
        = &\int \Wass_{h}\left(\mu_{Y(0)}^{z}, \mu_{Y(1)}^{z}\right) \diff \mu_Z(z) = V_{\cc}.
    \end{align*}
\end{theorem}
In the next section, we construct an adapted empirical distribution that converges under $\Wass_{\bc}$, which serves as the foundation for our adapted COT value estimator.

\section{Estimator of COT Value}\label{sec:main_result}
\subsection{Bernoulli Design}

\subsubsection{Adapted COT Value Estimator}\label{sec:estimator}
In this section, we introduce our adapted COT value estimator in the setting where the treatment is independent of both the observed outcome and covariates. This setting is related to a scenario referred to as the Bernoulli design (and also the randomized controlled trial) in causal inference, where $W_i$ are i.i.d.~Bernoulli with a constant probability parameter. In this setting, the sample from the control group, $((Y_i(0), Z_i(0)), i \in [n_0])$, can be viewed as an i.i.d.~sample drawn from $\mu_{Y(0), Z}$, with an analogous interpretation for the treated group.

\begin{assumption}[Bernoulli design]\label{a:sample_dist_random}
Assume that $((Y_i(w), Z_i(w)), i \in [n_w])$ are drawn i.i.d.~from $\mu_{Y(w), Z}$ for $w=0,1$.
\end{assumption}

Inspired by \cite{backhoff2022estimating}, we introduce the adapted empirical distribution for constructing our COT value estimator. First, we introduce the map below for discretization.
\begin{definition}[Cell-center projection]\label{defi:adapt_emp}
    Let \( r > 0\). For \( n \geq 1 \), partition the cube \( [0,1]^{d_Z} \) into a disjoint union of \( n^{r d_Z} \) smaller cubes, each with edge length \( n^{-r} \). Define \( \varphi_r^n: [0,1]^{d_Z} \to [0,1]^{d_Z} \) as the mapping that assigns each point to the center of the small cube it resides in. 
\end{definition}

\begin{definition}[Adapted empirical distribution]\label{defi:dist_est}
The adapted empirical distributions are defined as
\begin{align}\label{eq:mmu}
    \mmu_{Y(w), Z(w), n_w} &= \frac{1}{n_w}\sum_{i=1}^{n_w} \delta_{Y_i(w), \varphi_r^{n_w}(Z_i(w))}
\end{align}
for $w=0,1$, where $\varphi_r^n$ is defined in Definition \ref{defi:adapt_emp}.
\end{definition}
Compared to \cite{backhoff2022estimating}, our construction of $\mmu$ has a key distinction: we discretize only the covariate \(Z\), not the outcome \(Y\), since discretization is required only for the variables being conditioned on.

If $z \mapsto \mu_{Y(w)}^{z}$ is continuous under the weak topology, the adapted empirical distributions \eqref{eq:mmu} provide consistent approximations to the conditional distributions $\mu_{Y(w)}^z, w=0,1$, thereby mitigating the issue that each sample of $Z(0)$ or $Z(1)$ is associated with only a single observed outcome $Y$ when $\mu_Z$ has a density. Further, \cite{backhoff2022estimating} suggests that $(\mmu_{Y(w), Z(w), n_w}, w = 0,1)$ converges to their weak limits under $\Wass_{\bc}$, which is aligned with the conditions introduced in Theorem~\ref{prop:con_adapted_topology}. 

Note that larger discretization cells contain more observations, and the resulting local distribution estimator is associated with smaller variance but larger bias. We can choose the cell size to balance the bias-variance trade-off, where the optimal cell size depends on factors including the sample size. As the sample size may differ between treatment and control groups, for example, when units are assigned to treatment with relatively low probability, we discretize each group separately to allow for distinct optimal cell sizes.
In addition, in typical causal inference datasets, the covariate data for the treatment and control groups are not perfectly matched, also leading to a potential mismatch in the supports of the adapted empirical distributions in~\eqref{eq:mmu}. To address this issue, we match the discretized covariates \( Z(0) \) and \( Z(1) \) that may differ in value. This matching is formalized as a coupling with marginals $\mmu_{Z(0), n_0}$ and $\mmu_{Z(1), n_1}$.
Based on this coupling, we define our adapted COT value estimator as follows.
\begin{definition}[Adapted COT value estimator]\label{defi:cot_est}
    Given a coupling $\hat \pi_{n_0, n_1} \in \Pi(\mmu_{Z(0), n_0}, \mmu_{Z(1), n_1})$, the COT value estimator is defined as
    \begin{align*}
       \hat V_{n_0, n_1} &\Let \int \Wass_{h}(\mmu^{z_0}_{Y(0), n_0}, \mmu^{z_1}_{Y(1), n_1}) \diff \hat \pi_{n_0,n_1}(z_0, z_1).
    \end{align*}
\end{definition}

The COT value estimator achieves consistency as the sample size grows, provided that the coupling $\hat \pi_{n_0,n_1}$ converges to a joint distribution that exactly matches the limiting marginals.
\begin{theorem}[Consistency]\label{thm:consist}
    Under Assp.~\ref{a:compact_domain}-\ref{a:sample_dist_random}, if as $n_0,n_1 \rightarrow \infty$, $\hat \pi_{n_0,n_1}$ weakly converges to $(\mathrm{id}, \mathrm{id})_{\#}\mu_Z$ almost surely, then $\hat V_{n_0,n_1}$ converges to $V_{\cc}$ almost surely.
\end{theorem}

\subsubsection{Convergence Rate Analysis}
Next, we provide the finite-sample convergence rate for the adapted COT value estimator. We introduce the following assumptions.

\begin{assumption}[Smooth objective]\label{a:h_func}
    Assume that $h$ is $L_h$-Lipschitz continuous: for $\forall y_0, y_0', y_1, y_1' \in \calY$,
    \[
        \left| h(y_0, y_1) - h(y_0', y_1') \right| \leq L_h\left(\lnorm{y_0 - y_0'}{2} + \lnorm{y_1 - y_1'}{2}\right).
    \]
\end{assumption}

\begin{assumption}[Lipschitz kernel]\label{a:Lip_kernel}
    Assume that there is a constant $L_{Z} > 0$ such that for $\textup{$\mu_{Z}$-a.e.}~z_0, z_1 \in \calZ$,
    \begin{align*}
        &\Wass_1\left(\mu_{Y(w)}^{z_0}, \mu_{Y(w)}^{z_1}\right) \leq L_{Z} \lnorm{z_0 - z_1}{2} \ \  w = 0,1.
    \end{align*}
\end{assumption}
Assp.~\ref{a:h_func}-\ref{a:Lip_kernel} are aligned with standard regularity conditions in nonparametric estimation theory (e.g., \cite[Theorem 5.65]{wasserman2006all}).

\begin{assumption}[Coupling gap]\label{a:coupling_gap}
    Assume that $\forall n_0, n_1$ $:\EE[\int \lnorm{z_0 - z_1}{2} \diff \hat \pi_{n_0,n_1}(z_0, z_1)] \leq C_{\hat \pi}(n_0\wedge n_1)^{-r}$, where $C_{\hat \pi}$ is a constant independent of $n_0,n_1$.
\end{assumption}
Assp.~\ref{a:coupling_gap} enables quantifying the discrepancy between the estimated coupling \( \hat{\pi}_{n_0,n_1} \) and the identity coupling \( (\mathrm{id}, \mathrm{id})_{\#}\mu_Z \), and to ensure that this error aligns with the granularity of the discretization, as controlled by the parameter \( r \) in the map \( \varphi_r^n \).

\begin{theorem}[Finite-sample complexity]\label{thm:finite_sample}
Under Assp.~\ref{a:compact_domain}-\ref{a:coupling_gap} with $r = \frac{1}{d_Z + 2 \vee d_Y}$, we have 
    \[
    \EE\left[|\hat V_{n_0,n_1} - V_{\cc}|\right] \leq 
    \bar C \gamma_{d_Y, d_Z}(n_0 \wedge n_1),
    \]
    where
    \[
        \gamma_{d_Y, d_Z}(N) \Let 
    \begin{cases}
         N^{-\frac{1}{d_Z + 2 \vee d_Y}} \log(N) & \textup{if}\  d_Y \neq 2,\\
        N^{-\frac{1}{d_Z + 2 \vee d_Y}} & \textup{if}\  d_Y = 2.
    \end{cases}
    \]
    Here, $\bar C$ is a constant that depends on $d_Z, d_Y$, $L_h$ (Assp.~\ref{a:h_func}), and $L_Z$ (Assp.~\ref{a:Lip_kernel}).
\end{theorem}

One natural choice of $\hat \pi_{n_0, n_1}$ that satisfies Assp.~\ref{a:coupling_gap} is an optimal coupling of $\mmu_{Z(0), n_0}$ and $\mmu_{Z(1), n_1}$:
\begin{equation}\label{eq:pi_choice_ot}
    \hat \pi_{n_0,n_1} \in \argmin_{\pi \in \Pi(\mmu_{Z(0), n_0}, \mmu_{Z(1), n_1})} \EE_{\pi}[\|Z(0) - Z(1)\|_2].
\end{equation}

An additional advantage of the coupling \( \hat{\pi}_{n_0,n_1} \)~\eqref{eq:pi_choice_ot} is that it enjoys a robustness guarantee against potential data perturbations to the covariate distributions.

\begin{theorem}[Robustness]\label{rmk:robust}
    Define $\hat \pi_{n_0,n_1}$ by \eqref{eq:pi_choice_ot}. Suppose that the sample of $Z(w)$ is i.i.d.~drawn from $\mu^{\prime}_{w}$, satisfying \( \Wass_1(\mu_Z, \mu^{\prime}_{w}) \leq \epsilon \) for $w=0,1$. Then, 
    \[
        \EE\left[|\hat V_{n_0,n_1} - V_{\cc}|\right] \leq 2L_h L_Z \epsilon + \bar C \gamma_{d_Y, d_Z}(n_0\wedge n_1).
    \]
\end{theorem}

\subsection{Covariate-Dependent Treatment}\label{sec:propensity_score}

Unlike Bernoulli design where treatment assignment is independent of covariates by design, observational studies typically exhibit covariate-dependent treatment assignment. This results in non-constant propensity scores
$
e(z) = \PP(W_i = 1 \mid Z_i = z),
$
reflecting heterogeneity in treatment likelihood across covariate profiles. In this section, we discuss how the adapted COT value estimator can be tailored to this setting.

Under the unconfoundedness assumption (Assp.~\ref{assu:unconfoundedness}), the conditional distributions of potential outcomes given covariates remain identifiable, i.e., consistent with \( \mu_{Y(w)}^z \) for \( w = 0,1 \). However, since the treatment assignment probability depends on covariates, the marginal distribution of covariates may differ between the treatment group and the control group, inducing a covariate shift. 
To formalize this setting with covariate shift, we consider the following super-population model of $Y(0), Y(1), Z, W$ following the joint distribution \( \mu = \mu_{Y(0), Y(1), Z, W} \).
\begin{assumption}[Covariate shift]\label{a:covariate_shift}
    Assume that $((Z_i(w), Y_i(w)), i \in [n_w])$ are drawn i.i.d.~from $\mu_{Z | W=w} \otimes \mu_{Y(w)}^Z$ for $w=0,1$.
\end{assumption}

Since our estimand of interest $V^*$ relies on the marginal covariate distribution, irrespective of treatment assignment, it is essential to correct for the distributional shift induced by non-constant propensity scores. Specifically, we reweight the samples to align with the marginal covariate distribution. 
By Proposition~\ref{prop:ID.marginal}, when the propensity score $e(z)$ is known, the appropriate population-level weighting is given by
\begin{align*}
    \xi_{w,i} = \frac{\mu_{Z}(Z_i(w))}{\mu_{Z \mid W = w}(Z_i(w))} \propto \frac{1}{1 - w -e(Z_i(w))}, ~ W_i = w.
\end{align*}
On finite samples, when the propensity score \( e(z) \) is unknown but estimated by \( \hat{e}(z) \), we adopt a group-wise self-normalized weight: for $w=0,1$, $\hat \xi_{w,G}$ is defined as
\begin{align}\label{eq:IPW}
    \frac{(1 - w-\hat e(\varphi_r^{n_w}(G)))^{-1}}{\sum_{G} (1 - w -\hat e(\varphi_r^{n_w}(G)))^{-1} |\{i: Z_i(w) \in G\}|} 
\end{align}
to ensure the weights sum to one. Here, $G \in \Phi_r^n$ is a cell corresponding to projection map $\varphi_r^n$, where $\Phi_r^n \Let \left\{ (\varphi_r^n)^{-1}(\{x\}): x \in \varphi_r^n([0,1]^{d_Z}) \right\}$, and $\varphi_r^n(G)$ is the cell center of $G\in \Phi_r^n$. Note that the weight \eqref{eq:IPW} is different from the common self-normalized weights that are based on the sample value of $Z$. This is used to adapt the construction of the adapted COT value estimator with discretization to the setting in Assp.~\ref{a:covariate_shift}.
Finally, our reweighted adapted empirical distributions are defined by: for $w=0,1$,
\[
\begin{aligned}
    \mmu_{Y(w), Z(w), n_w} &= \sum_{G \in \Phi_r^{n_w}} \sum_{i: Z_i(w) \in G} \hat \xi_{w,G} \delta_{ Y_i(w), \varphi_r^{n_w}(Z_i(w))}.  
\end{aligned}
\] 
We then still employ $\mmu_{Y(w), Z(w), n_w}, w=0,1$ along with $\hat \pi_{n_0,n_1}$ defined in Eq.~\eqref{eq:pi_choice_ot},
to construct the adapted COT value estimator as formulated in Definition~\ref{defi:cot_est}.

Specifically, we implement a two-fold cross-fitting procedure to ensure the independence between the estimated propensity score $\hat e$ and the adapted empirical distributions $\mmu$. We randomly split the data into two folds of approximately equal size. For each fold, we construct $\hat e$ using data from the other fold, and then construct $\mmu_{Y(w), Z(w), n_w}$ on the current fold with the estimated propensity score plugged in. This yields two estimates of the lower bound, and we further average them to form our final estimator $\hat V_{n_0,n_1}$.

We introduce the assumptions on the estimation error of \( \hat{e}(\cdot) \) to enable the convergence analysis of \( \hat{V}_{n_0,n_1} \).
\begin{assumption}[Lipschitz $\hat e$]\label{a:lip_propensity_score}
    Assume that there is a constant $L_e$, such that $|\hat e(z) - \hat e(z')| \leq L_e\lnorm{z- z'}{2}$ for any $z,z' \in \calZ^2$.
\end{assumption}
\begin{assumption}[Bounds of $\hat e$]\label{a:bound_e}
    Assume that there is a constant $\eta \in (0,\frac{1}{2})$, such that $\eta \leq \hat e(z) \leq 1 - \eta$ for any $z \in \calZ$.
\end{assumption}
\begin{assumption}[Average error of propensity score]\label{a:est_propensity_score} Assume that for $n_0,n_1 \geq 1$,
    \begin{align*}
        \EE\left[\frac{1}{n_w}\sum_{i=1}^{n_w} \left|n_w \hat \xi_{w,i}  -  \frac{\diff \mu}{\diff \mu_{|W=w}} (Z_i(w)) \right|\right] &\leq C_{\textup{w}} n_w^{-r},
    \end{align*}
    where $C_{\textup{w}}$ is a constant independent of $n_0, n_1$ and 
    \begin{align*}
        \hat \xi_{w,i} \Let \frac{(1 - w -\hat e(Z_i(w)))^{-1}}{\sum_{l=1}^n (1- w -\hat e(Z_l(w)))^{-1}}, w=0,1.
    \end{align*}
\end{assumption}

Assp.~\ref{a:est_propensity_score} characterizes the average estimation error of the estimated propensity score function \( \hat{e} \). Notably, Assp.~\ref{a:est_propensity_score} uses the sample-based self-normalized weight to be compatible with the common setting in the literature. Also, we impose a mild convergence rate requirement of order \( O(n^{-r}) \), where \( r \) will be chosen to be less than \( 1/d_Z \). The rate \( O(n^{-1/d_Z}) \) aligns with typical nonparametric convergence rates for nuisance parameter estimation. Assp.~\ref{a:lip_propensity_score} and~\ref{a:bound_e} impose boundedness and smoothness conditions on \( \hat{e}(\cdot) \). Particularly, the smoothness condition is necessary because we evaluate \( \hat{e} \) at the centers of discretized cells, as discussed in the following.

\begin{theorem}[Finite-sample complexity with covariate shift]\label{thm:finite_sampleII}
    Under Assp.~\ref{assu:overlap}-\ref{a:cont_obj}, 
    \ref{a:h_func}-\ref{a:est_propensity_score} with $r = \frac{1}{d_Z + 2 \vee d_Y}$, then
    \[
        \EE\left[|\hat V_{n_0,n_1} - V_{\cc}|\right] \leq C^{\dagger} \gamma_{d_Y, d_Z}(n_0\wedge n_1),
    \]
    where $C^{\dagger}$ is a constant that depends on $d_Z, d_Y$, 
    $\delta$ (Assp.~\ref{assu:overlap}), $L_h$ (Assp.~\ref{a:h_func}), $L_Z$ (Assp.~\ref{a:Lip_kernel}), 
    $L_e$ (Assp.~\ref{a:lip_propensity_score}), $\eta$ (Assp.~\ref{a:bound_e}) and $C_{\textup{w}}$ (Assp.~\ref{a:est_propensity_score}).
\end{theorem}

\section{Simulation}\label{sec:simulation}
\subsection{Selection of Cell Size}\label{sec:select_cellsize}
In Definition~\ref{defi:adapt_emp}, we design the discretization cells with edge length growing in the order of  \( n^{-r} \), resulting in \( O(n^{r d_Z}) \) small cells, in order to balance the bias-variance tradeoff with respect to the sample size $n$. In practice, the constant in $O(n^{r d_Z})$ may also affect the estimation accuracy. Specifically, in this section we set the number of cells to the nearest integer of $c n^{r d_Z}$.

We conduct a sensitivity analysis of our method with respect to $c$ under the quadratic location model (see Section~\ref{sec:sim_setting} (b)) with $n_0 = n_1 = 300$. As shown in Table~\ref{tab:sensitivity}, we compute the average relative estimation error (and their standard error), i.e., $\EE[|\hat V_{n_0, n_1} - V_{\cc}|] / V_{\cc}$ for $5$ choices of $c$, where $\EE$ is approximated through $500$ Monte Carlo simulation runs. Compared to the baseline $c = 1$, we can see that when $c$ is relatively small, increasing $c$ improves the estimation accuracy while its performance tends to stabilize as $c$ grows larger. 

\begin{center}
\vspace{-0.12in}
\captionof{table}{Sensitivity to $c$ (baseline $c=1$).}\label{tab:sensitivity}
\begin{tabular}{ccc}
\toprule
$c$ & Relative error & $\Delta$ vs. baseline \\
\midrule
0.6  & 0.21 (5.4e-3) & $+38.1 \% $ \\ 
0.8  & 0.17 (5.2e-3)  & $+13.4 \% $ \\ 
\textbf{1.0} & \textbf{0.15 (4.9e-3)}  & \textbf{0.0} \\ 
1.2  & 0.14 (4.8e-3) & $-6.4 \% $ \\ 
1.4  & 0.14 (4.9e-3) & $-7.8 \%$ \\ 
\bottomrule
\end{tabular}
\end{center}
In practice, since the true value of $V_{\cc}$ is unknown, we suggest the following approach to select a $c$ that achieves a relatively higher estimation accuracy: (i) given a dataset with size $N$, generate $B=50$ bootstrap samples, each with size $N$, (ii) compute the average $\hat V_{n_0, n_1}$ value of the bootstrap samples for each candidate of $c$, and plot a curve on the average values, (iii) pick the ``elbow'' point of the curve (the point with largest distance to the line connecting the start and end point of the curve), and choose the corresponding $c$ value. An illustration is shown in Figure~\ref{fig:elbow}. 
\begin{center}
\includegraphics[width=0.9\linewidth]{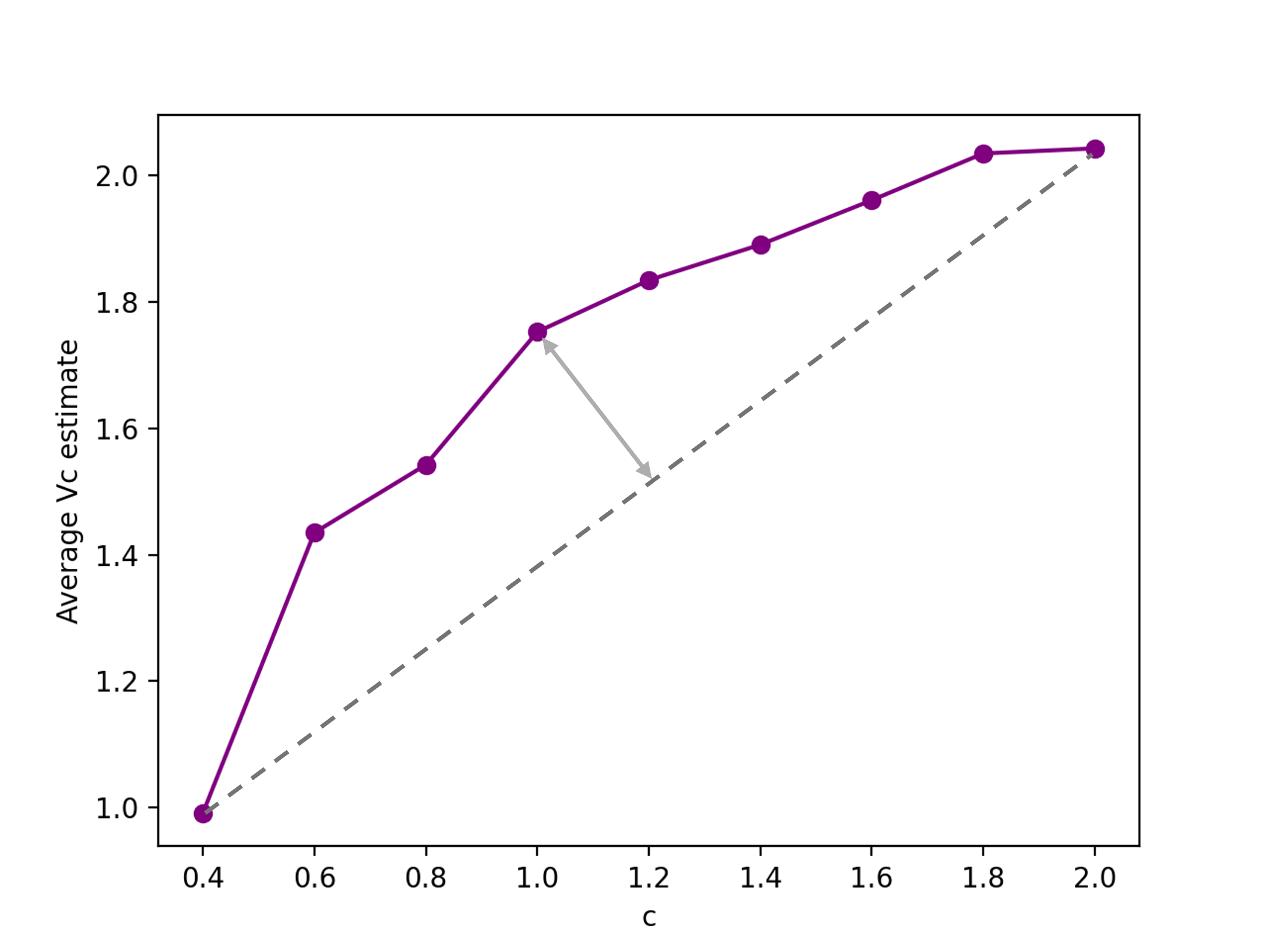}
\captionof{figure}{Selection of the elbow point.}\label{fig:elbow}
\end{center}

\begin{figure*}[!t]
  \centering
  \begin{subfigure}{0.32\textwidth}
    \centering
    \includegraphics[width=\linewidth]{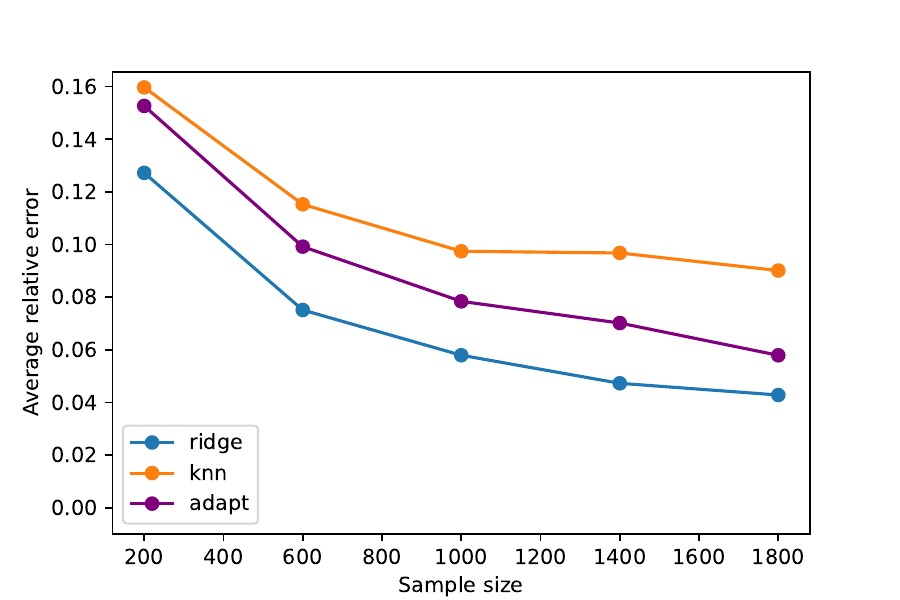}
    \subcaption{Linear location model\\(Bernoulli design)}\label{fig:a}
  \end{subfigure}\hfill
  \begin{subfigure}{0.32\textwidth}
    \centering
    \includegraphics[width=\linewidth]{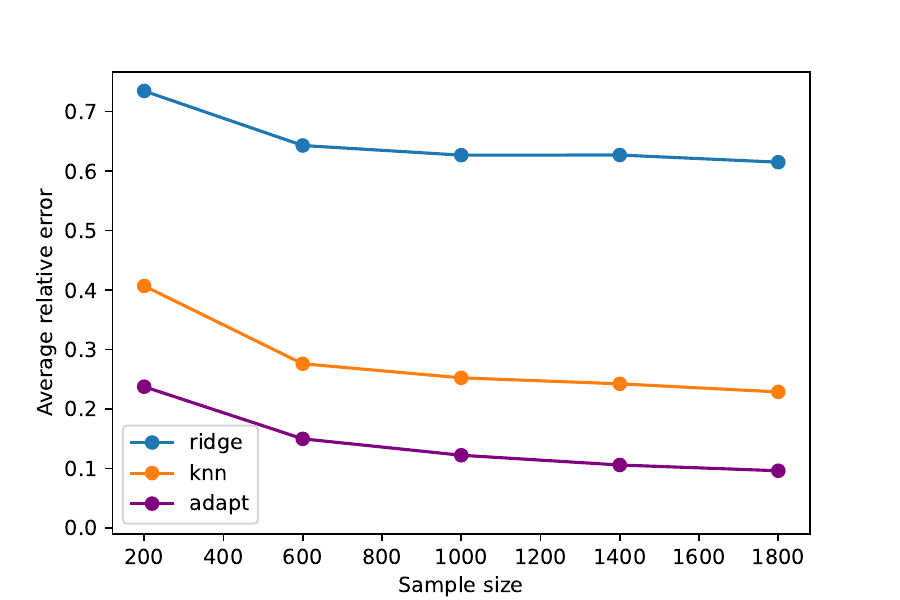}
    \subcaption{Quadratic location model\\(Bernoulli design)}\label{fig:b}
  \end{subfigure}\hfill
  \begin{subfigure}{0.32\textwidth}
    \centering
    \includegraphics[width=\linewidth]{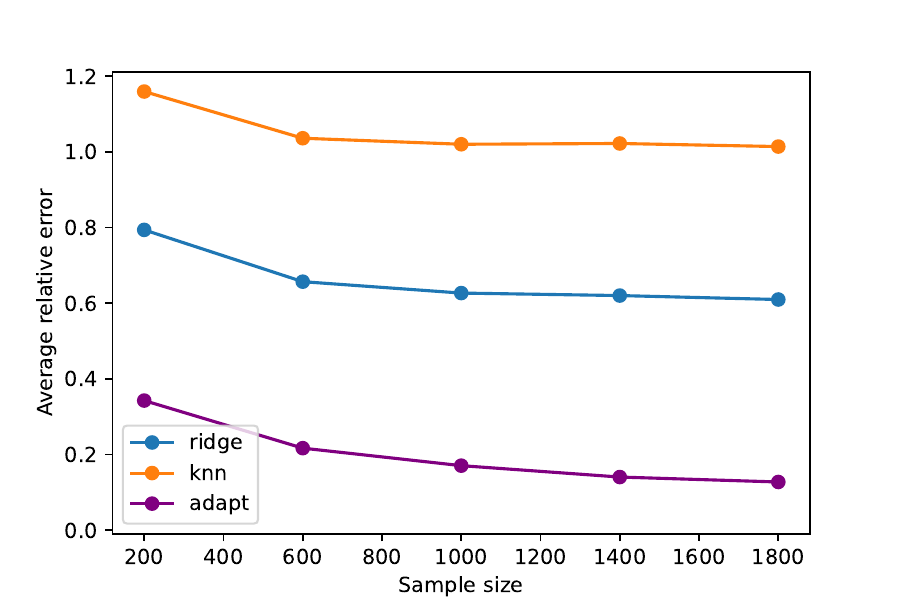}
    \subcaption{Scale model\\(Bernoulli design)}\label{fig:c}
  \end{subfigure}

  
  \begin{subfigure}{0.32\textwidth}
    \centering
    \includegraphics[width=\linewidth]{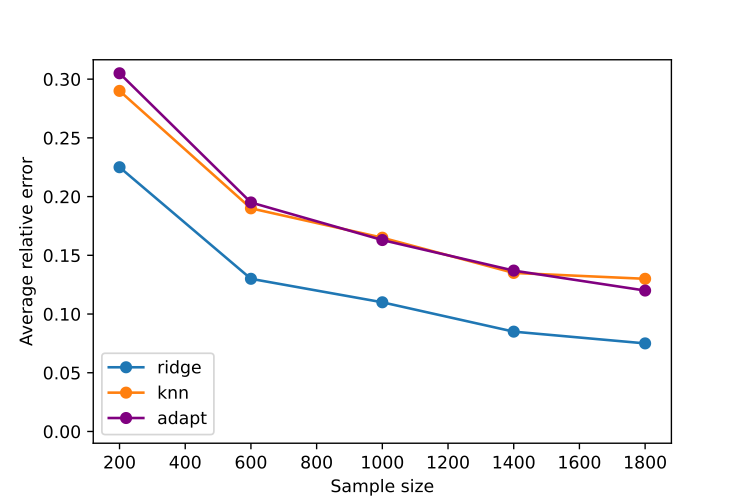}
    \subcaption{Linear location model\\(Covariate-dependent)}\label{fig:a2}
  \end{subfigure}\hfill
  \begin{subfigure}{0.32\textwidth}
    \centering
    \includegraphics[width=\linewidth]{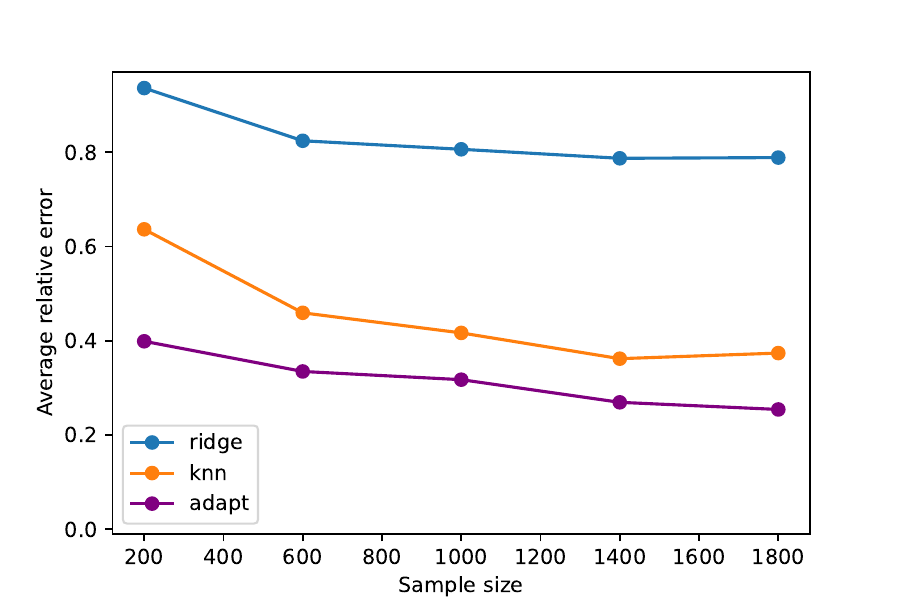}
    \subcaption{Quadratic location model\\(Covariate-dependent)}\label{fig:b2}
  \end{subfigure}\hfill
  \begin{subfigure}{0.32\textwidth}
    \centering
    \includegraphics[width=\linewidth]{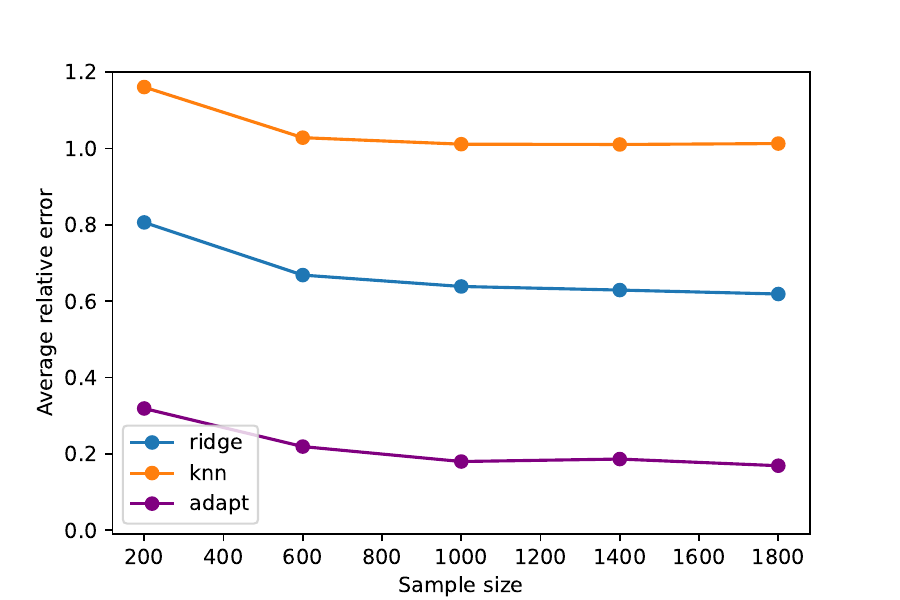}
    \subcaption{Scale model\\(Covariate-dependent)}\label{fig:c2}
  \end{subfigure}
  \caption{Estimation accuracy comparison between the adapted COT estimator and \texttt{DualBounds}. In the legend,  ``adapt'' refers to the adapted COT estimator, ``ridge'' refers to the ridge regression–based \texttt{DualBounds}, and ``knn'' refers to the KNN–based \texttt{DualBounds}. The average relative error is computed over 500 Monte Carlo repetitions. For the uncertainty, we compute the standard error of the mean (SEM), i.e. $\textup{SD} / \sqrt{500}$, where $\textup{SD}$ is the standard deviation. The numerical values of SEM satisfy: (a): $< 0.005$; (b)(c): $< 0.02$; (d)(e)(f): $<0.018$. }
  \label{fig:three-wide2}
\end{figure*}
\begin{table*}[!t]
\centering
\caption{Average relative error (SEM) over 500 Monte Carlo repetitions on the STAR-based semi-synthetic data.}
\label{tab:STAR}
\begin{tabular}{c|cccc}
\toprule
Method ($N = n_0 + n_1 = 2n_0$) & $N=400$ & $N=800$ & $N=1200$ & $N=1600$ \\
\midrule
ridge regression–based \texttt{DualBounds} & $4.0\text{e}{-2}$ (1.3e-3) & $3.2\text{e}{-2}$ (1.0e-3) & $2.9\text{e}{-2}$ (0.9e-3) & $2.7\text{e}{-2}$ (0.8e-3) \\
KNN–based \texttt{DualBounds}   & $5.1\text{e}{-2}$ (1.6e-3) & $4.5\text{e}{-2}$ (1.3e-3) & $4.5\text{e}{-2}$ (1.1e-3) & $4.4\text{e}{-2}$ (1.0e-3) \\
\textbf{Our method} & $\mathbf{3.6e\text{-}2 
 ~(1.2e\text{-}3)}$  & $\mathbf{2.9e\text{-}2 
 ~(1.0e\text{-}3)}$  & $\mathbf{2.6e\text{-}2 
 ~(0.8e\text{-}3)}$  & $\mathbf{2.5e\text{-}2 
 ~(0.8e\text{-}3)}$  \\
\bottomrule
\end{tabular}
\end{table*}

\subsection{Simulation Results}\label{sec:sim_setting}
\textbf{Synthetic data}. To evaluate the effectiveness of the adapted COT value estimator in simulation\footnote{The code and additional experimental results are available at 
\href{https://github.com/siruilin1998/causalPIviaCOT}
{\texttt{github.com/siruilin1998/causalPIviaCOT}}.}, we introduce the following synthetic data-generating mechanism. For $w=0,1$, $Y(w) = F_{m}(f_w(Z), \varepsilon_w)$, where $Z, \varepsilon_w, Y(w)$ are independent with each other. The model $F_m$ includes: (i) location model $F_{1}(u,v) = u + v$, and (ii) scale model $F_2(u,v) = (u_j v_j, j \in [d_Y])$. More concretely, we consider $d_Y = d_Z = 1$, $Z, \varepsilon_w \sim N(0,1)$, and three models:
\begin{itemize}
    \item [(a)] Linear location model ($m=1$) with $f_0(z) = -0.6 z, f_1(z) = 1.6z$.
    \item [(b)] Quadratic location model ($m=1$) with $f_0(z) = -0.2 z^2, f_1(z) = 0.6 z ^2$.
    \item [(c)] Scale model ($m=2$) with $f_0(z) = 0.5 z - 0.35, f_1(z)=1.1z + 0.35$.
\end{itemize}  
For these Gaussian models, the $V_{\cc}$'s population value has a closed form without referring to an optimization problem (see appendix for a proof), and is computable through a simple Monte Carlo algorithm with parametric convergence rate. Therefore, these models enable the evaluation of the estimation accuracy. 

We compare the estimation accuracy of our adapted COT value estimator with the method in \cite{ji2023model}, implemented in their Python package \texttt{DualBounds}\footnote{\url{https://dualbounds.readthedocs.io}}. Their approach relies on estimating nuisance functions associated with the covariate--outcome model. In our comparison, we evaluate both the ridge regression-based and \(k\)-nearest neighbor (KNN)-based \texttt{DualBounds} methods. Our implementation is based on the Python Optimal Transport (\texttt{POT}\footnote{\url{https://pythonot.github.io/}}) library \cite{flamary2021pot}. 

In Figure~\ref{fig:three-wide2}, we present the comparison in the setting of: (i) Bernoulli design ((a)(b)(c)) and (ii) Covariate-dependent treatment ((d)(e)(f)) with $e(z) = (1 + e^{-\frac{3z}{2}})^{-1}\ \forall z \in \RR$, which we assume to be fully known for all methods under evaluation. The figure shows that: (i) For the linear location model, the ridge-based \texttt{DualBounds} method performs best due to its use of a linear nuisance estimator that aligns with the true model structure. Nonetheless, the adapted COT value estimator achieves comparable accuracy despite not relying on model-specific assumptions (see (a)). (ii) For the quadratic location model, the adapted COT value estimator outperforms both versions of \texttt{DualBounds} and the KNN-based outperforms the ridge-based, as the latter cannot capture the nonlinearity in the true model. (iii) For the scale model, the adapted COT estimator also has the best performance.  

\textbf{Real data}. We also consider the real-world dataset from the Student Achievement and Retention (STAR) Demonstration Project \cite{angrist2009incentives}, where the academic performance (outcome) is measured by the GPA recorded in the first academic year and the baseline GPA serves as the covariate. To make the COT value accessible at the population level, we evaluate the methods using hypothetical data generated from a model fitted to the real data. Table~\ref{tab:STAR} presents the accuracy on estimating the PI bounds of the correlation between two potential outcomes, and our adapted COT value estimator consistently outperforms the others.

\section{Conclusion and Discussion}
\paragraph{Conclusion.} In this paper, we study the optimal covariate-assisted partial identification sets for causal estimands by solving COT problems.
Our finite-sample COT estimator avoids nuisance function estimation, does not require well-specified models, and enjoys a provable statistical convergence rate. Furthermore, our adapted COT value estimator answers the question proposed in the discussion of \cite{lin2025tightening} on how to combine the adapted Wasserstein distance into the COT value estimation. Specifically, our adapted estimator is their $V_{\text{causal}}(\eta)$ when $\eta$ approached infinity.

\paragraph{Triangular transport maps.}
We next discuss the connection between our COT framework and triangular transport maps. When the transport cost is quadratic and the conditional outcome distributions admit densities, the optimal coupling can be represented by a triangular transport map \cite{carlier2010knothe,carlier2016vector,hosseini2023conditional}. Such maps arise from a sequential alignment of conditional distributions and coincide, in the one-dimensional outcome case, with the classical Knothe--Rosenblatt rearrangement \cite{rosenblatt1952remarks,knothe1957contributions}. 

Triangular transport has been used in several causal modeling and inference frameworks \cite{charpentier2023optimal,de2024transport,balakrishnan2023conservative}, where it is often interpreted as a structural mechanism for generating counterfactual outcomes. Our perspective is complementary. Rather than modeling unit-level counterfactuals via an estimated transport map, we use the induced causal coupling structure to define and estimate a distributional functional. In particular, our estimator targets the COT value directly, without requiring explicit construction of the underlying triangular map.

Nevertheless, we conjecture that our discretization approach (Section~\ref{sec:estimator}) could be adapted to recover a triangular transport map in a manner analogous to plug-in constructions in the unconditional setting; see, for example, the one-nearest-neighbor estimator in \cite[Section~4.2]{manole2024plugin}. In that case, the estimation error of the induced triangular transport gap would be expected to match the order of the estimation error for the COT value itself.

\paragraph{Multiple treatment levels.} We discuss here the possibilities of an extension to settings with multiple treatment levels. Suppose the treatment variable takes more than two levels. A natural extension of our framework is to consider multi-marginal optimal transport, which constructs a joint coupling across all conditional outcome distributions given the covariates. This provides a geometric framework for jointly aligning multiple potential outcome distributions. The unconditional multi-treatment setting has been studied in \cite{gao2024bridging}.

At the same time, the use of transport maps in counterfactual modeling has important limitations. Recent work \cite{dance2025counterfactual,de2025good} shows that optimal transport maps generally cannot serve as structural models for counterfactual assignments when more than two treatment levels are present. We emphasize that our methodology does not interpret OT maps as structural counterfactual mechanisms. Instead, OT and COT are employed as distributional comparison tools: they measure discrepancies between conditional outcome laws and construct PI intervals via geometric constraints on admissible couplings. Consequently, conditional OT remains meaningful in multi-treatment settings for constructing PI intervals and bounding functionals of potential outcomes.

\clearpage
\subsubsection*{Acknowledgements}
The material in this paper is based upon work supported by the Air Force Office of Scientific Research under award number FA9550-20-1-0397. Additional support is gratefully acknowledged from NSF 2118199, 2229012, 2312204, 2403007, and ONR 13983111.

\bibliographystyle{plain}
\bibliography{ref}




\clearpage
\section*{Checklist}

\begin{enumerate}

  \item For all models and algorithms presented, check if you include:
  \begin{enumerate}
    \item A clear description of the mathematical setting, assumptions, algorithm, and/or model. [Yes]
    \item An analysis of the properties and complexity (time, space, sample size) of any algorithm. [Yes]
    \item (Optional) Anonymized source code, with specification of all dependencies, including external libraries. [Yes]
  \end{enumerate}

  \item For any theoretical claim, check if you include:
  \begin{enumerate}
    \item Statements of the full set of assumptions of all theoretical results. [Yes]
    \item Complete proofs of all theoretical results. [Yes]
    \item Clear explanations of any assumptions. [Yes]     
  \end{enumerate}

  \item For all figures and tables that present empirical results, check if you include:
  \begin{enumerate}
    \item The code, data, and instructions needed to reproduce the main experimental results (either in the supplemental material or as a URL). [Yes]
    \item All the training details (e.g., data splits, hyperparameters, how they were chosen). [Yes]
    \item A clear definition of the specific measure or statistics and error bars (e.g., with respect to the random seed after running experiments multiple times). [Yes]
    \item A description of the computing infrastructure used. (e.g., type of GPUs, internal cluster, or cloud provider). [Yes]
  \end{enumerate}

  \item If you are using existing assets (e.g., code, data, models) or curating/releasing new assets, check if you include:
  \begin{enumerate}
    \item Citations of the creator If your work uses existing assets. [Yes]
    \item The license information of the assets, if applicable. [Yes]
    \item New assets either in the supplemental material or as a URL, if applicable. [Yes]
    \item Information about consent from data providers/curators. [Yes]
    \item Discussion of sensible content if applicable, e.g., personally identifiable information or offensive content. [Yes]
  \end{enumerate}

  \item If you used crowdsourcing or conducted research with human subjects, check if you include:
  \begin{enumerate}
    \item The full text of instructions given to participants and screenshots. [Yes]
    \item Descriptions of potential participant risks, with links to Institutional Review Board (IRB) approvals if applicable. [Yes]
    \item The estimated hourly wage paid to participants and the total amount spent on participant compensation. [Yes]
  \end{enumerate}

\end{enumerate}

\clearpage
\input{main_supplet}

\end{document}

%% file: sty.tex
\usepackage{url}
\usepackage{amsthm}
\usepackage{dsfont}
\usepackage{bm}
\usepackage{graphicx}
\usepackage{subcaption}
\usepackage{booktabs}
\usepackage{hyperref}

\newtheorem{theorem}{Theorem}
\newtheorem{lemma}{Lemma}
\newtheorem{assumption}{Assumption}
\newtheorem{proposition}{Proposition}
\newtheorem{corollary}{Corollary}
\newtheorem{example}{Example}

\theoremstyle{definition}
\newtheorem{definition}{Definition}
\newtheorem{remark}{Remark}

\def\EE{\mathbb{E}}

\def\PP{\mathbb{P}}
\def\QQ{\mathbb{Q}}
\def\RR{\mathbb{R}}

\def\calD{\mathcal{D}}

\def\calG{\mathcal{G}}

\def\calL{\mathcal{L}}

\def\calN{\mathcal{N}}

\def\calP{\mathcal{P}}

\def\calX{\mathcal{X}}
\def\calY{\mathcal{Y}}
\def\calZ{\mathcal{Z}}

\newcommand{\cc}{\textup{c}}

\usepackage{amssymb}
\newcommand{\Let}{\triangleq}

\def\1{\mathbbm{1}}

\newcommand\independent{\protect\mathpalette{\protect\independenT}{\perp}}
\def\independenT#1#2{\mathrel{\rlap{$#1#2$}\mkern2mu{#1#2}}}

\newcommand{\argmin}{\mathop{\mathrm{argmin}}}

\usepackage{booktabs}
\usepackage{adjustbox}
\usepackage{multirow}
\usepackage{listings}

\usepackage{xspace}

\newcommand{\diff}{{\rm d}}

\newcommand{\lnorm}[2]{\left\|{#1} \right\|_{{#2}}}

\def\independenT#1#2{\mathrel{\rlap{$#1#2$}\mkern2mu{#1#2}}}

\usepackage{enumitem}

\usepackage{accents}

\newcommand{\bc}{\textup{a}}
\newcommand{\mmu}{\hat{\bm{\mu}}}

\def\EE{\mathbb{E}}

\def\PP{\mathbb{P}}
\def\QQ{\mathbb{Q}}
\def\RR{\mathbb{R}}

\def\calD{\mathcal{D}}

\def\calG{\mathcal{G}}

\def\calL{\mathcal{L}}

\def\calN{\mathcal{N}}

\def\calP{\mathcal{P}}

\def\calX{\mathcal{X}}
\def\calY{\mathcal{Y}}
\def\calZ{\mathcal{Z}}

\newcommand{\Wass}{\mathds W}

\def\1{\mathbbm{1}}

\def\independenT#1#2{\mathrel{\rlap{$#1#2$}\mkern2mu{#1#2}}}

\usepackage{booktabs}
\usepackage{adjustbox}
\usepackage{multirow}
\usepackage{listings}

\usepackage{tikz}

\usepackage{xspace}

\def\independenT#1#2{\mathrel{\rlap{$#1#2$}\mkern2mu{#1#2}}}

\definecolor{myblue}{rgb}{.8, .8, 1}
\definecolor{mathblue}{rgb}{0.2472, 0.24, 0.6} 
\definecolor{mathred}{rgb}{0.6, 0.24, 0.442893}
\definecolor{mathyellow}{rgb}{0.6, 0.547014, 0.24}

\usepackage{enumitem}

%% file: main_supplet.tex
\clearpage
\appendix
\thispagestyle{empty}

\onecolumn
\aistatstitle{Supplementary Materials}

\noindent\textbf{Organization}. The appendix is organized as follows.
Section~\ref{sec:details_realdata} presents the details of the real data experiment discussed in Section~\ref{sec:sim_setting}.
Section~\ref{sec:supp_experiment} presents supplementary experiments to show the robustness and effectiveness of our method.
Section~\ref{sec:technical_background} reviews technical background knowledge. 
Section~\ref{sec:proof_population} provides proofs for the results in Section~\ref{sec:problem_form}.
Section~\ref{sec:proof_basic} provides proofs for the results in Section~\ref{sec:connection}.
Section~\ref{sec:proof_finite_sample} provides proofs for the results in Section~\ref{sec:main_result}.
Section~\ref{sec:proof_lemma} collects proofs for supporting lemmas.

\section{More Details on the Real Data Experiment}\label{sec:details_realdata}
\begin{figure}[ht]
    \centering
    \includegraphics[width=1\linewidth]{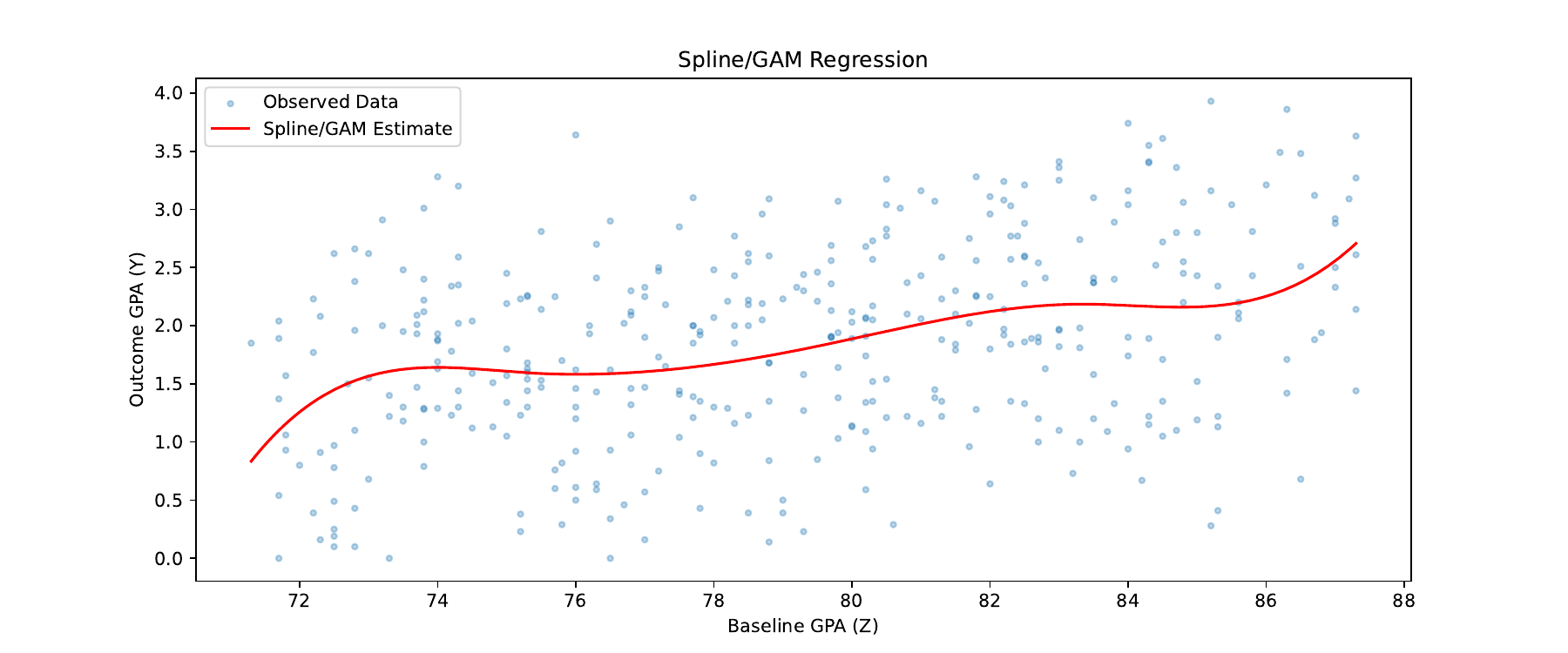}
    \caption{Cubic spline regression of the outcome variable $Y$ (GPA) versus the covariate $Z$ (baseline GPA) for the treatment group. The distribution of $Z$ is modeled using Gaussian kernel density estimation (KDE), while the relationship between the outcome $Y$ and $Z$ is modeled using cubic spline regression. The Wasserstein distance between the empirical and KDE-generated distributions of $Z$ is less than $0.2$. The Wasserstein distance between the observed and model-generated distributions of $Y$ is less than $0.07$ for the treatment group and less than $0.09$ for the control group.}
    \label{fig:treat_fit}
\end{figure}

In this section, we present more details of the real data experiment in Section~\ref{sec:sim_setting}, which is based on the Student Achievement and Retention (STAR) Demonstration Project \cite{angrist2009incentives}, an initiative designed to assess the impact of scholarship incentive programs on academic performance. Experiments were run on a MacBook Air (Apple M3, 8-core CPU, 16 GB RAM), without GPU acceleration. 

In the STAR study, the treatment is access to a scholarship incentive and was randomly assigned.
Academic performance is measured by the GPA recorded at the end of the first academic year.
In addition to the treatment indicator and outcome variable, the dataset includes baseline GPA (measured prior to treatment assignment), which serves as a key covariate due to its strong predictive influence on academic outcomes.

We are interested in the correlation between two potential outcomes of each unit $i$, which is defined as
\[
\rho \Let \frac{\mathbb{E}[Y_i(1) Y_i(0)] - \EE[Y_i(1)] \EE[Y_i(0)]}{\sqrt{\mathrm{Var}(Y_i(0)) \mathrm{Var}(Y_i(1))}}.
\]  
The unidentifiable part of $\rho$ arises from the estimand $\mathbb{E}[Y_i(1) Y_i(0)]$. Then, to obtain the corresponding $V_{\cc}$ value for $\rho$, we could utilize the cost function $h(y_0, y_1) = (y_0 + y_1)^2$. To ensure the PI bounds are accessible at the population level, we generate hypothetical data from a model fitted to the real data.
Particularly, we fit the real data using the model
$Y(w) = f_w(Z(w)) + \calN(0, \sigma_w^2)\ w=0,1,$
and subsequently generate data based on the estimated functions $f_w$ and noise level $\sigma_w$.

To ensure the true (population-level) PI bounds are accessible, we generate hypothetical data from a model fitted to the real data.
Particularly, we fit the real data using the model
\[
    Y(k) = f_k(Z(k)) + \calN(0, \sigma_k^2)\ \ \ \  k=0,1.
\]
and subsequently generate data based on the estimated functions $f_k$ and noise level $\sigma_k$ for $k=0,1$ (see Figure~\ref{fig:treat_fit}). In addition, we model the distribution of baseline GPA ($Z$) using Gaussian kernel density estimation.
Specifically, we consider
\[
    f_k(Z) = \sum_{j=1}^6 \hat \beta_{kj} \phi_j(Z)\ \ \ \ k=0,1,
\]
where $(\phi_j, j=1,...,6)$ are cubic spline basis functions. The regression coefficients $(\hat \beta_{kj}, j=1,...,6, k=0,1)$ are estimated via ridge regression applied to the real data, with the regularization parameter selected through cross-validation.

\section{Supplementary Experiments}\label{sec:supp_experiment}
\subsection{Robustness to Covariate Distributional Shift}
\begin{figure}[ht]
    \centering
    \begin{minipage}{0.5\textwidth}
    \centering
    \includegraphics[width = 1 \textwidth]{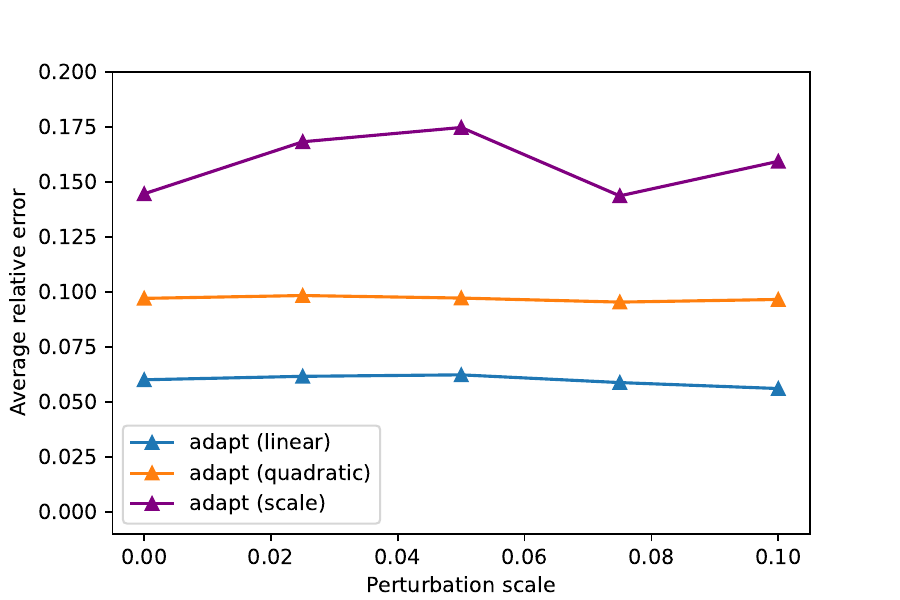}
    \end{minipage}
    \hfill   

    \caption{
    Robustness of our method to covariate mismatch.  
    The average relative error is defined as \( \mathbb{E}\left[ \left| \text{estimator} - V_{\cc} \right| \right] / V_{\cc} \), where the expectation \( \mathbb{E} \) is approximated by averaging over 500 Monte Carlo repetitions. To quantify the uncertainty, we compute the standard error of the mean (SEM) as $\textup{SD} / \sqrt{500}$, where $\textup{SD}$ denotes the standard deviation of the relative errors across repetitions. The maximum SEM values are less than 0.006 in this setting.}
    \label{fig:experiment_robust}
\end{figure}

In this section, we consider a setting in which the covariates of the treatment group are shifted. Let $\mu_{Z(0)}, \mu_{Z(1)}$ denote the (marginal) generating distribution of $Z(0), Z(1)$. Specifically, we set \( \mu_Z = \mu_{Z(0)} = \mathcal{N}(0,1) \) and perturb the treatment covariate distribution as \( \mu_{Z(1)} = \mathcal{N} (0,1) + \eta \varepsilon \), where $\varepsilon$ is an independent $\calN(0,1)$ noise and \( \eta \) controls the magnitude of the perturbation. Using the same three models described in Section~\ref{sec:sim_setting}, we evaluate the estimation accuracy of the adapted COT value estimator under varying levels of \( \eta \). The results are reported in Figure~\ref{fig:experiment_robust}. We observe that the adapted COT value estimator remains stable across small perturbation levels (\( \eta \in [0, 0.1] \)), illustrating its robustness to covariate distribution shift (see also Theorem~\ref{rmk:robust}). Intuitively, this robustness stems from the projection of covariates onto a finite set of representative values: small perturbations in the covariates are often absorbed during coarsening, as they rarely alter cell assignments. Moreover, the use of optimal coupling to match the discretized marginals helps mitigate the impact of systematic perturbations such as location shifts.

\subsection{Comparison with a Fr\'echet--Hoeffding-Type Approach}
When the outcomes are one-dimensional, the copula models \cite{nelsen2006introduction}, and the Fréchet–Hoeffding bounds \cite{heckman1997making, manski1997mixing, fan2010sharp, kawakami2025moments} can provide partial identification bounds in closed forms without solving an OT-type optimization problem. Particularly, for the case $d_Y = 1$ and $h(y_0, y_1) = (y_0 - y_1)^2$, Hoeffding-type results yield a closed-form optimizer for \eqref{eq:COT.lower}.
(Note that Hoeffding-type bounds do not provide consistent partial identification sets for $d_Y>1$, whereas our procedure does.)

\begin{proposition}[Fr\'echet--Hoeffding bounds for $h(y_0, y_1) = (y_0 - y_1)^2$]
Let $F_w(\cdot | z)\ w=0,1$ be the marginal CDFs of $Y_w$ conditioned on $Z =z$ on $\mathbb{R}$ 
with finite second moments, and denote their quantile functions by 
$F_w^{-1}(u | z) = \inf\{x : F_w(x | z) \ge u\}$ for $w=0,1$.
Then, for any coupling $\pi$ of $(Y_0, Y_1, Z)$ with these marginals,
\begin{align*}
\int_{\calZ} \int_0^1 \!\big(F_0^{-1}(u | z) - F_1^{-1}(u | z)\big)^2 \, \diff u \diff z
&\;\le\;
\mathbb{E}_{\pi}\!\big[(Y_0 - Y_1)^2\big]
\;\le\;
\int_{\calZ} \int_0^1 \!\big(F_0^{-1}(u | z) - F_1^{-1}(1-u | z)\big)^2 \, \diff u\diff z.
\end{align*}
Conditioned on $Z = z$, the lower bound is attained by the comonotone coupling 
$Y_0 = F_0^{-1}(U | z)$, $Y_1 = F_1^{-1}(U | z)$, 
and the upper bound by the countermonotone coupling 
$Y_0 = F_0^{-1}(U | z)$, $Y_1 = F_1^{-1}(1-U | z)$,
where $U \sim \mathrm{Uniform}(0,1)$.
\end{proposition}

Based on the above result, a Fr\'echet-Hoeffding-Type approach \cite{aronow2014sharp, balakrishnan2023conservative} for estimating $V_{\cc}$ proceeds as follows:
\begin{enumerate}[label=(\roman*)]
    \item \textbf{Quantile regression:} 
    Within the control and treatment groups, estimate the conditional quantile functions by running quantile regressions. For example, in the following, we use a non-parametric quantile regression method, random forest quantile regression (python package: quantile\_forest \cite{Johnson2024}\footnote{https://pypi.org/project/quantile-forest/}).
    
    \item \textbf{Fr\'echet--Hoeffding bound:} 
    For each covariate value (in either the control or treatment group), use the estimated conditional quantiles to compute the squared loss.
    
    \item \textbf{Aggregation:} 
    Approximate the overall value by taking the sample average of the values of the conditional squared loss derived in (ii) over the distinct covariate values.
\end{enumerate}

To compare our adapted COT value estimator and the above Fr\'echet--Hoeffding--based approach, we consider the scale model (Section~\ref{sec:sim_setting} (c)) with the same model parameter as there. We run a Monte Carlo simulation with 200 repetitions. The mean absolute error (standard error) results are collected in Table~\ref{tab:compare_FH} (In this setting we could compute the oracle true value (=1.92) of the PI bounds with high accuracy). The results show that our method is comparable with the Fr\'echet--Hoeffding--based method in this special one-dimensional outcome setting.

\begin{table}[h!]
\centering
\caption{Comparison of our method and Fr\'echet--Hoeffding--based method using the scale model.}\label{tab:compare_FH}
\begin{tabular}{lcccc}
\toprule
 & N = 100 & N = 500 & N = 1000 & N = 1500 \\
\midrule
\textbf{Our Method} & 0.356 (0.025) & 0.176 (0.011) & 0.138 (0.007) & 0.122 (0.007) \\
F--H Based & 0.276 (0.019) & 0.180 (0.009) & 0.156 (0.007) & 0.148 (0.006) \\
\bottomrule
\end{tabular}
\end{table}

\section{Technical Background}\label{sec:technical_background}
\subsection{Regular Kernel}\label{app:disintegration}
In this paper, we always assume that joint distribution like $\mu_{Y,Z}$ that can be written as $\mu_{Y,Z} = \mu_{Z} \otimes \mu_Y^Z$, i.e. for any measurable function $g$, $\int g(y,z) \diff \mu = \int_{\calZ} \int_{\calY} g(y,z) \diff \mu_Y^z(y) \diff \mu_Z(z)$, with a regular kernel $\mu_Y^z$. The definition of a regular kernel is stated as follows. 

\begin{definition}[Regular kernel]
    A kernel $\mu_Y^z$ is regular if (i) for any $z \in \calZ$, $\int \mathbf{1}(y \in \cdot)\diff  \mu_Y^z (y)$ is a probability measure, and (ii) for any measurable set $B \subseteq \calY$, $z \mapsto \int \mathbf{1}(y \in B)\diff  \mu_Y^z (y)$ is a measurable function.
\end{definition}

We can always assume that $\mu$ has a regular kernel due to the following result.
\begin{proposition}[Disintegration {\cite[Theorem 5.3, 5.4]{kallenberg1997foundations}}]
    If $\calY \subseteq \RR^{d_Y}$ and is equipped with the Borel $\sigma$-algebra, then for a measure $\mu_{Y,Z} \in \calP(\calY \times \calZ)$, there is a regular kernel $\mu_Y^z$ such that $\int g(y,z) \diff \mu = \int_{\calZ} \int_{\calY} g(y,z) \diff \mu_Y^z(y) \diff \mu_Z(z)$ for any measurable function $g$. Further, $\mu_Y^z$ is unique $\mu_Z$-a.e.
\end{proposition}

\subsection{Wasserstein vs. Adapted Wasserstein: A Distance Metric Comparison}\label{sec:adapted_wass}
This section presents the two metrics employed to study the continuity of COT.

\begin{definition}[Wasserstein distance]\label{defi:wass}
    The $1$-Wasserstein distance between $\mu, \nu \in \calP(\calX)$ is defined as
    \begin{align*}
        \Wass_1(\mu, \nu) \Let \min_{\pi \in \Pi(\mu, \nu)} \EE_{\pi}[\lnorm{X(0) - X(1)}{2}].
    \end{align*}
\end{definition}

The $1$-Wasserstein distances possess key properties that will be instrumental in our proof.
\begin{lemma}[Triangle inequality]\label{lem:tri_ineq}
    $\Wass_1(\mu, \mu') \leq \Wass_{1}(\mu, \nu) + \Wass_{1}(\nu, \mu')$.
\end{lemma}
\begin{lemma}[Convexity]\label{lem:convex}
    The map $\mu \mapsto \Wass_1(\mu, \mu')$ is convex.
\end{lemma}

For $\calX = \calZ \times \calY$, the so-called adapted Wasserstein distance induces a topology that is typically stronger than the weak topology, which can be metrized by the Wasserstein distances.
\begin{definition}[Adapted Wasserstein distance]
    The adapted Wasserstein distances between $\mu, \nu \in \calP(\calX)$ is defined as
    \begin{align*}
        \Wass_{\bc}(\mu, \nu) \Let \min_{\pi \in \Pi_{\bc}(\mu,\nu)} \EE_{\pi}[\lnorm{X(0) - X(1)}{2}],
    \end{align*}
    where $X(0) = (Z(0), Y(0)), X(1) = (Z(1), Y(1))$, and
    \begin{align*}
        \Pi_{\bc}(\mu,\nu) \Let& \left\{\pi \in \Pi(\mu,\nu): \text{under}~\pi,~ Z(1) \perp\!\!\!\perp Y(0) \text{ conditioned on } Z(0), \right. \\
        &\  \left. Z(0) \perp\!\!\!\perp Y(1) \text{ conditioned on } Z(1)\right\}.
    \end{align*}
\end{definition}

The following proposition is stated to maintain consistency with Definition~\ref{defi:adapted_wassI} introduced earlier. 
\begin{proposition}[Reformulation of adapted Wasserstein distance; see, e.g. {\cite[Proposition 5.2]{backhoff2017causal}}]
    Let $\mu = \mu_Z \otimes \mu_Y^{Z}$ and $\nu = \nu_Z \otimes \nu_Y^{Z}$, we have
    \[
        \Wass_{\bc}(\mu, \nu) = \min_{\pi \in \Pi(\mu_{Z}, \nu_Z)} \int \lnorm{z - z'}{2} + \Wass_{1}(\mu_Y^{z}, \nu_Y^{z'})~\diff \pi(z, z').
    \]
\end{proposition}
Since $\Pi_{\bc}(\mu, \nu) \subseteq \Pi(\mu, \nu)$, clearly $\Wass_1 \leq \Wass_{\bc}$, but the reverse may not hold. Consider an i.i.d.~sample $(X_i = (Z_i, Y_i), i \in [n])$ drawn from a distribution $\mu$, and let $\hat \mu_n \Let \frac{1}{n}\sum_{i=1}^n\delta_{X_i}$. It is well known that $\Wass_1(\hat \mu_n, \mu)$ converges to zero almost surely as $n$ goes to infinity. However, as shown in \cite[Proposition 1]{pflug2016empirical}, $\Wass_{\bc}(\hat \mu_n, \mu)$ may fail to converge to zero, and a remedy to ensure convergence is to replace $\hat \mu_n$ with its convolution with a suitable kernel function. The intuition behind this approach is to ensure that the kernel of the convoluted empirical distribution becomes smooth in the variable it is conditioned on.
This idea aligns with the findings of \cite{blanchet2024bounding}, which show that for distributions with Lipschitz kernels, the weak topology and the topology induced by the adapted Wasserstein distance are equivalent.

However, since the convolved empirical distribution is not discrete, its practical implementation requires an additional sampling step for computation. To address this, \cite{backhoff2022estimating} propose a discrete alternative---the adapted empirical distribution---which mitigates the convergence issue by assigning the sample \( (X_i)_{i \in [n]} \) to a finite number of cells, which is the approach we adopt in this paper.

\subsection{Fournier and Guillin's Result}\label{sec:rate.wasserstein}
This section provides the convergence rate of the empirical distribution under the Wasserstein distance, as established by \cite{fournier2015rate}.

Suppose that the sample $(X_i, i \in [n])$ i.i.d.~follows $\mu$ and denote $\hat \mu_n = \frac{1}{n}\sum_{i=1}^n \delta_{X_i}$. Theorem 1 in \cite{fournier2015rate} characterizes a nearly sharp convergence rate for empirical Wasserstein distances.
\begin{theorem}[Empirical Wasserstein convergence rate, Theorem 1 \cite{fournier2015rate}]\label{thm:fournier.guillin}
    Let $\mu \in \mathcal{P}(\mathbb{R}^d)$ and let $q > 0$. Assume that $M_k(\mu) := \int \lnorm{x}{2}^k \diff \mu < \infty$ for some $k > q$. There exists a constant $C$ depending only on $q, d, k$ such that, for all $n \geq 1$,
    \[
    \EE\big( \Wass_q(\hat \mu_n, \mu)^q \big) 
    \leq CM_k^{q/k}(\mu)
    \begin{cases}
    n^{-1/2} + n^{-(k-q)/k} & \text{if } q > d/2, k \neq 2q, \\[6pt]
    n^{-1/2} \log(1+n) + n^{-(k-q)/k} & \text{if } q = d/2, k \neq 2q, \\[6pt]
    n^{-q/d} + n^{-(k-q)/k} & \text{if } q \in (0, d/2), k \neq d/(d-q),
    \end{cases}
    \]
    where $ \Wass_q(\mu, \nu) \Let \min_{\pi \in \Pi(\mu, \nu)} \EE_{\pi}[\lnorm{X(0) - X(1)}{2}^q]^{\frac{1}{q}}$.
\end{theorem}

In this work, we will utilize the following corollary.
\begin{corollary}[$k > 2$]\label{cor:founier}
    Under the same conditions of Theorem~\ref{thm:fournier.guillin}, if $k > 2$, then
    \[
        \EE[\Wass_1(\hat \mu_n, \mu)] \leq C R_{d}(n),
    \]
    where $C$ is a constant that depends on $k, q, d, \EE_{\mu}[\lnorm{x}{2}^k]$, and 
    \begin{align}\label{eq:R}
        R_{d}(n)=
    \begin{cases}
    n^{-\frac{1}{2 \vee d}} & \text{if } d \neq 2, \\[6pt]
    n^{-\frac{1}{2}} \log(3+n) & \text{if } d = 2.
    \end{cases}
    \end{align}
\end{corollary}
Note that here we take $\log(3+n)$ rather than $\log(1+n)$ to make $x \mapsto x R_{d}(x)$ a concave function. 


\section{Proofs of Section \ref{sec:problem_form}}\label{sec:proof_population}
\subsection{Proof of Proposition~\ref{prop:ID.marginal}}
\begin{proof}[Proof of Proposition~\ref{prop:ID.marginal}]

We prove the equality for the treatment potential outcome ($k=1$); the result for the control potential outcome ($k=0$) follows analogously.
In fact,
\begin{align*}
    \PP(Y_i(1) = y, Z_i = z)
    &= \frac{\PP(Y_i(1) = y, Z_i = z, W_i = 1)}{\PP(W_i = 1 \mid Y_i(1) = y, Z_i = z)} \\
    &= \frac{\PP(Y_i(1) = y, Z_i = z \mid W_i = 1) \cdot \PP(W_i = 1)}{\PP(W_i = 1 \mid Z_i = z)}\\
    &=  \PP(Y_i(1) = y, Z_i = z \mid W_i = 1) \cdot \frac{\PP(W_i = 1)}{e(z)},
    \end{align*}
where we use Assumption~\ref{assu:unconfoundedness} in the second last equality.
\end{proof}

\subsection{Proof of Proposition~\ref{prop:recursive_form}}
\begin{remark}
In this proof, we will use the notions of Borel measurability, lower semianalyticity, and universal measurability, which are discussed in detail in \cite[Section 7.7]{bertsekas1978stochastic}.
\end{remark}

\begin{proof}[Proof of Proposition~\ref{prop:recursive_form}]
    We first show that the right-hand side (RHS) $\int \Wass_h(\mu_{Y(0)}^z, \mu_{Y(1)}^z) \diff \mu_{Z}(z)$ is well-defined.

    Since $h$ is a measurable function bounded from below, then the map
    \[
        \calL_h \Let (z, \gamma) \mapsto \int h(y_0, y_1) \diff \gamma(y_0, y_1)
    \]
    is a Borel (measurable) and thus lower semianalytic (l.s.a.). 

    Note that the maps $z \mapsto \mu_{Y(0)}^z$ and $z \mapsto \mu_{Y(1)}^z$ are both Borel since $\mu_{Y(0)}^z, \mu_{Y(1)}^z$ are  regular kernels (see Section~\ref{app:disintegration}). Further, $\{(p,q,\gamma): \gamma \in \Pi(p,q)\}$ is weakly closed. Then, the set
    \[
        D \Let \{(z, \gamma): \gamma \in \Pi(\mu_{Y(0)}^z, \mu_{Y(1)}^z)\}
    \]
    is Borel.

    As a result, $\Wass_h(\mu_{Y(0)}^z, \mu_{Y(1)}^z)$ is the infimum of the function $\calL_h$ over the fiber of $D$ at $z$. By \cite[Proposition 7.47]{bertsekas1978stochastic}, the map  $z \mapsto \Wass_h(\mu_{Y(0)}^z, \mu_{Y(1)}^z)$ is l.s.a. and in particular universally measurable. Therefore, the RHS $\int \Wass_h(\mu_{Y(0)}^z, \mu_{Y(1)}^z) \diff \mu_{Z}(z) + \epsilon.$ is well-defined.

    Next, we prove the equation from two directions as follows.

    \textbf{LHS $\leq$ RHS}:
    For any $\epsilon > 0$ By \cite[Proposition 7.50(b)]{bertsekas1978stochastic}, there is a universally measurable $\epsilon$-optimizer for 
    \[
        \inf_{\gamma \in \Pi(\mu_{Y(0)}^z, \mu_{Y(1)}^z)} \int h(y_0, y_1) \diff \gamma(y_0, y_1).
    \]
    That is, there exists a universally measurable map $z \mapsto \gamma^z_{\epsilon} \in \Pi(\mu_{Y(0)}^z, \mu_{Y(1)}^z)$ such that
    \begin{equation}\label{eq:eps_opt}
        \int h(y_0, y_1) \diff \gamma_{\varepsilon}^{z}(y_0, y_1) \leq \inf_{\gamma \in \Pi(\mu_{Y(0)}^z, \mu_{Y(1)}^z)} \int h(y_0, y_1) \diff \gamma(y_0, y_1) + \varepsilon.
    \end{equation}
    Then, apply \cite[Proposition 7.45]{bertsekas1978stochastic}, we can construct a measure $\pi_{\epsilon} \in \calP(\calY^2 \times \calZ)$ such that
    \[
        \int g(y_0, y_1, z) \diff \pi_{\epsilon}(y_0, y_1, z) = \int \int_{\calY^2} g(y_0, y_1, z) \diff \gamma_{\varepsilon}^{z}(y_0, y_1) \diff \mu_{Z}(z)
    \]
    for any measurable function $g$.
    
    As a result, $\pi \in \Pi_{\cc}$, and thus
    \[
        V_{\cc} = \min_{_{\pi \in \Pi_{\cc}}} \int h(y_0, y_1) \diff \pi \leq \int h(y_0, y_1) \diff \pi_{\epsilon} \leq \int \Wass_h(\mu_{Y(0)}^z, \mu_{Y(1)}^z) \diff \mu_{Z}(z) + \epsilon,
    \]
    where the second inequality is due to \eqref{eq:eps_opt}.

    Send $\epsilon \rightarrow 0^+$, we get $V_{\cc} \leq \int \Wass_h(\mu_{Y(0)}^z, \mu_{Y(1)}^z) \diff \mu_{Z}(z)$.
    
    \textbf{LHS $\geq$ RHS}: For any $\pi \in \Pi_{\cc}$, it can be written as $\pi = \mu_Z \otimes \pi_{Y(0), Y(1)}^{Z}$, and $\pi_{Y(0), Y(1)}^{z} \in \Pi(\mu_{Y(0)}^z, \mu_{Y(1)}^z)$ $\mu_Z$-almost everywhere. This is because for any measurable function $g$,
    \[
        \int g(y, z) \diff \mu_{Y(0), Z}(y, z) = \int g(y, z) \diff \pi_{Y(0), Z}(y, z) = \int_{\calZ} \int_{\calY} g(y_0, z) \diff \pi_{Y(0)}^z(y) \diff \mu_{Z}(z),
    \]
    and the argument on $(Y(1), Z)$ is similar.

    As a result, 
    \[
        \begin{aligned}
            \int h(y, y') \diff \pi(y,y',z) 
            = &  \int_{\calZ} \diff \mu_Z(z) \int_{\calY \times \calY} h(y, y') \diff \pi_{Y(0), Y(1)}^z(y,y') \\
            \geq & \int_{\calZ} \Wass_{h} (\mu_{Y(0)}^z, \mu_{Y(1)}^z) \diff \mu_Z(z).
        \end{aligned}
    \]
    Taking infimum on $\pi$ completes this part of proof.

\end{proof}

\section{Proofs of Section~\ref{sec:connection}}\label{sec:proof_basic}

\subsection{Proof of Theorem~\ref{thm:optimal_causalbounds}}
\begin{lemma}[$\Pi_{\cc}$ is compact]\label{lem:Pic_compact}
    $\Pi_{\cc}$ is a compact subset of $\calP(\calY^2 \times \calZ)$ under the weak topology.
\end{lemma}

\begin{proof}[Proof of Theorem~\ref{thm:optimal_causalbounds}]
    Given the function $h$, the estimand $V_{\cc}$ only depends on $\mu_{Y(0), Z}, \mu_{Y(1), Z}$, which are identifiable from the observed samples, thus (i) holds true. 

    For (ii), by definition of $V_{\cc}$, for any $\pi \in \Pi_{\cc}$, $\EE_{\pi}[h(Y(0), Y(1))] \geq V_{\cc}$, indicating that $V_{\cc}$ is a valid lower bound for $V^*$. Further, since $h$ is bounded and lower semicontinuous, $\EE_{\pi}[h(Y(0), Y(1))]$ is a lower semicontinuous function in $\pi$ with respect to the weak topology (\cite[Lemma 4.3]{villani2009optimal}). By Lemma~\ref{lem:Pic_compact}, $\Pi_{\cc}$ is compact, thus the minimum $V_{\cc}$ can be attained. 

    For (iii), for any $v = (1-\theta) V_{\cc} + \theta \tilde V_{\cc}$, where $\theta \in [0,1]$, there is a coupling $\pi_{v} = (1-\theta) \pi_{\cc} + \theta \tilde \pi_{\cc}$ such that $\pi_{v} \in \Pi_{\cc}$ (since $\Pi_{\cc}$ is a convex set) and $v = \EE_{\pi_{v}}[h(Y(0), Y(1))]$. Then, the PI set is convex and thus is an interval. 

    Since now $h$ is continuous, then by (ii), $\tilde V_{\cc}$ is the tightest upper bound for the PI set. As a result, the PI set is equal to $[V_{\cc}, \tilde V_{\cc}]$.
\end{proof}

\begin{remark}[Unbounded $h$]
    The results of Theorem~\ref{thm:optimal_causalbounds} can be extended to possibly unbounded function $h$ such that, there are functions $a(y) \in L^1(\mu_{Y(0)}), b(y) \in L^1(\mu_{Y(1)})$ ($f \in  L^1(\mu)$ denotes that $\int |f| \diff \mu < \infty$) and $|h(y_0, y_1)| \leq a(y_0) + b(y_1)$. For the proof, we can consider $\tilde h = h(y_0, y_1) - (a(y_0) + b(y_1))$ for the lower bound and $\tilde h' = a(y_0) + b(y_1) - h(y_0, y_1)$ for the upper bound.
\end{remark}

\subsection{Proof of Theorem~\ref{prop:con_adapted_topology}}

\begin{lemma}[Uniformly continuous kernel]\label{lemm:unicont_kernel}
    Suppose that $\calY, \calZ$ are both compact and $g_w(z) = \mu_{Y(w)}^{z}: \calZ \rightarrow \calP(\RR^{d_{Y}}) \ w=0,1$ are continuous under the weak topology. Then, for any $\delta > 0$, there is a constant $C(\delta)$ such that $\Wass_1\left(\mu_{Y(k)}^{z_0}, \mu_{Y(k)}^{z_1}\right) \leq \delta + C(\delta)\|z_0 - z_1\|_2$.
\end{lemma}

\begin{lemma}[Lipschitz continuity of $\Wass_h$]\label{lemm:LipWassh}
    For any $\epsilon \geq 0$, suppose that $h: \calY\times\calY \rightarrow \RR$ satisfies that for any $y, y' \in \calY$,
    \begin{align}\label{eq:hsmoothassp}
        \left| h(y) - h(y') \right| \leq L_h \lnorm{y - y'}{2} + \epsilon.
    \end{align}
    Then, for any $(\mu, \nu, \mu', \nu') \in \calP(\calY)$,
    \[  
        \left| \Wass_{h}(\mu, \nu) - \Wass_{h}(\mu', \nu') \right| \leq L_h  \left(\Wass_1(\mu, \mu') + \Wass_1(\nu, \nu')\right) + 2\epsilon.
    \]
\end{lemma}
\begin{proof}[Proof of Theorem~\ref{prop:con_adapted_topology}]
    For $\nu = \nu_{Y(0), Y(1), Z}$, we bound the following gap:
    \[
        \begin{aligned}
            & \left|V_{\cc} - \int \Wass_{h}\left(\nu_{Y(0)}^{z}, \nu_{Y(1)}^{z}\right) \diff \nu_{Z}(z)\right|\\
            = & \left|\int \Wass_{h}\left(\mu_{Y(0)}^{z}, \mu_{Y(1)}^{z}\right) \diff \mu_{Z}(z) - \int \Wass_{h}\left(\nu_{Y(0)}^{z}, \nu_{Y(1)}^{z}\right) \diff \nu_{Z}(z)\right|\\ 
            \leq & \underbrace{\left|\int \Wass_{h}\left(\mu_{Y(0)}^{z}, \mu_{Y(1)}^{z}\right) \diff \mu_{Z}(z) - \int \Wass_{h}\left(\mu_{Y(0)}^{z}, \nu_{Y(1)}^{z}\right) \diff \nu_{Z}(z)\right|}_{\text{(Term A)}}\\ 
            & +  \underbrace{\left|\int \Wass_{h}\left(\mu_{Y(0)}^{z}, \nu_{Y(1)}^{z}\right) \diff \nu_{Z}(z) - \int \Wass_{h}\left(\nu_{Y(0)}^{z}, \nu_{Y(1)}^{z}\right) \diff \nu_{Z}(z) \right|}_{\text{(Term B)}}.
        \end{aligned}
    \]
    Here, the first equation is due to Proposition~\ref{prop:recursive_form}.

    By Assp.~\ref{a:compact_domain}-\ref{a:cont_obj}, $h$ is continuous and supported on a compact domain, then $h$ is uniformly continuous. Therefore, for any $\epsilon > 0$, there is a constant $\widetilde C(\epsilon)$ such that for any $y, y' \in \calY$,
    \[
        |h(y) - h(y')| \leq \widetilde C(\epsilon)\lnorm{y - y'}{2} + \epsilon.
    \]
    Thus, Lemma~\ref{lemm:LipWassh} is applicable and for any $(\mu, \nu, \mu', \nu') \in \calP(\calY)$,
    \begin{align}\label{eq:eqA1}
        \left| \Wass_{h}(\mu, \nu) - \Wass_{h}(\mu', \nu') \right| \leq \widetilde C(\epsilon) \left(\Wass_1(\mu, \mu') + \Wass_1(\nu, \nu')\right) + 2\epsilon.
    \end{align}

    Similarly, by Lemma~\ref{lemm:unicont_kernel}, we have
    \begin{align}\label{eq:eqA2}
        \Wass_1\left(\mu_{Y(k)}^{z_0}, \mu_{Y(k)}^{z_1}\right) \leq  C(\epsilon)\|z_0 - z_1\|_2 + \epsilon.
    \end{align}

    For \textbf{(Term A)}: we highlight that we consider the optimal coupling $\pi \in \calP(\calY^2 \times \calZ^2)$ that is attained in the adapted Wasserstein distance $\Wass_{\bc}(\mu_{Y(1), Z}, \nu_{Y(1), Z})$ to connect $\mu_{Z}$ and $\nu_{Z}$, then for $\epsilon, \epsilon' > 0$,
    \[
        \begin{aligned}
            (\text{Term A}) = &\left|\int \Wass_{h}\left(\mu_{Y(0)}^{z}, \mu_{Y(1)}^{z}\right) \diff \mu_{Z}(z) - \int \Wass_{h}\left(\mu_{Y(0)}^{z}, \nu_{Y(1)}^{z}\right) \diff \nu_{Z}(z)\right|\\
            = & \left|\int \Wass_{h}\left(\mu_{Y(0)}^{z}, \mu_{Y(1)}^{z}\right) -  \Wass_{h}\left(\mu_{Y(0)}^{z'}, \nu_{Y(1)}^{z'}\right) \diff \pi_{Z, Z'}(z, z')\right|\\
            \leq & \int \left|\Wass_{h}\left(\mu_{Y(0)}^{z}, \mu_{Y(1)}^{z}\right) -  \Wass_{h}\left(\mu_{Y(0)}^{z'}, \nu_{Y(1)}^{z'}\right) \right|\diff \pi_{Z, Z'}(z, z')\\
            \leq & \widetilde  C(\epsilon) \int \Wass_{1}\left(\mu_{Y(0)}^{z}, \mu_{Y(0)}^{z'}\right) + \Wass_{1}\left(\mu_{Y(1)}^{z}, \nu_{Y(1)}^{z'}\right) \diff \pi_{Z, Z'}(z, z') + 2\epsilon\\
            \leq & \widetilde C(\epsilon) \int (C(\epsilon') \|z - z'\|_2 + \epsilon') + \Wass_{1}\left(\mu_{Y(1)}^{z}, \nu_{Y(1)}^{z'}\right) \diff \pi_{Z, Z'}(z, z') + 2 \epsilon\\
            \leq & \widetilde C(\epsilon) (C(\epsilon') \vee 1) \Wass_{\bc}(\mu_{Y(1), Z}, \nu_{Y(1), Z}) + 2 \epsilon + \widetilde C(\epsilon) \epsilon'.
        \end{aligned}
    \]
    Here, the second inequality is due to \eqref{eq:eqA1}; the third inequality is due to \eqref{eq:eqA2}; the last inequality is due to the definition of $\pi$, which achieves the \textit{optimality} in $\Wass_{\bc}(\mu_{Y(1), Z}, \nu_{Y(1), Z})$.

    For \textbf{(Term B)}: we have
    \[
        \begin{aligned}
            (\text{Term B}) = & \left|\int \Wass_{h}\left(\mu_{Y(0)}^{z}, \nu_{Y(1)}^{z}\right) \diff \nu_{Z}(z) - \int \Wass_{h}\left(\nu_{Y(0)}^{z}, \nu_{Y(1)}^{z}\right) \diff \nu_{Z}(z) \right|\\
            \leq & \int \left| \Wass_{h}\left(\mu_{Y(0)}^{z}, \nu_{Y(1)}^{z}\right)  - \Wass_{h}\left(\nu_{Y(0)}^{z}, \nu_{Y(1)}^{z}\right) \right| \diff \nu_{Z}(z)\\
            \leq & \widetilde  C(\epsilon) \int \Wass_1(\mu_{Y(0)}^{z}, \nu_{Y(0)}^{z}) \diff \nu_Z(z) + 2\epsilon,
        \end{aligned}
    \]
    where the last inequality is due to \eqref{eq:eqA1}.

    Now, we consider the optimal coupling $\tilde \pi \in \calP(\calY^2 \times \calZ^2)$ that is attained in $\Wass_{\bc}(\mu_{Y(0), Z}, \nu_{Y(0), Z})$, then
    \begin{align}\label{eq:marginal_to_couple}
    \begin{split}
            &\int \Wass_1(\mu_{Y(0)}^{z}, \nu_{Y(0)}^{z}) \diff \nu_Z(z)\\
            =&\int \Wass_1(\mu_{Y(0)}^{z'}, \nu_{Y(0)}^{z'}) \diff \tilde \pi(z,z')\\
            =&\int \Wass_1(\mu_{Y(0)}^{z'}, \nu_{Y(0)}^{z'}) - \Wass_1(\mu_{Y(0)}^{z'}, \mu_{Y(0)}^{z}) \diff \tilde \pi(z,z') + \int \Wass_1(\mu_{Y(0)}^{z'}, \mu_{Y(0)}^{z}) \diff \tilde \pi(z,z')\\
            \leq & \int \Wass_1(\nu_{Y(0)}^{z'}, \mu_{Y(0)}^{z}) \diff \tilde \pi(z,z') + \int \Wass_1(\mu_{Y(0)}^{z'}, \mu_{Y(0)}^{z}) \diff \tilde \pi(z,z')\\
            \leq & \int \Wass_1(\nu_{Y(0)}^{z'}, \mu_{Y(0)}^{z}) \diff \tilde \pi(z,z') + C(\epsilon') \int \lnorm{z - z'}{2} \diff \tilde \pi(z,z') + \epsilon'\\
            \leq & (C(\epsilon') \vee 1) \Wass_{\bc}(\mu_{Y(0), Z}, \nu_{Y(0), Z}) + \epsilon'. 
    \end{split}
    \end{align}
    Here, the first inequality is due to the triangle inequality of the metric $\Wass_1$; the second inequality is due to \eqref{eq:eqA2}; the last inequality is due to the definition of $\tilde \pi$.
    
    Therefore, we get
    \[
        (\text{Term B}) \leq \widetilde C(\epsilon) (C(\epsilon') \vee 1) \Wass_{\bc}(\mu_{Y(1), Z}, \nu_{Y(1), Z}) + 2 \epsilon + \widetilde C(\epsilon) \epsilon'.
    \]

    Finally, combine (Term A) and (Term B), and we get
    \[
        \begin{aligned}
            &\left|V_{\cc} - \int \Wass_{h}\left(\nu_{Y(0)}^{z}, \nu_{Y(1)}^{z}\right) \diff \nu_{Z}(z)\right| \\
            \leq & \widetilde C(\epsilon) (C(\epsilon') \vee 1) (\Wass_{\bc}(\mu_{Y(0), Z}, \nu_{Y(0), Z}) + \Wass_{\bc}(\mu_{Y(1), Z}, \nu_{Y(1), Z})) + 2(2 \epsilon + \widetilde C(\epsilon) \epsilon').
        \end{aligned}  
    \]
    Replace $\nu$ by the measures from the sequence of $(\nu_n, n \geq 1)$ and send $\epsilon', \epsilon \rightarrow 0$, we get the desired result.
\end{proof}

\section{Proofs of Section~\ref{sec:main_result}}\label{sec:proof_finite_sample}
\textbf{Notation.} In this section and the following, to simplify the notation, we will let $n_0 = n, n_1 = m$, and $\xi_0 = w, \xi_1 = \xi$.

\subsection{Proof of Theorem~\ref{thm:finite_sample}}
The proof of Theorem~\ref{thm:finite_sample} requires several useful results from \cite{backhoff2022estimating}. Since our definition of adapted empirical distribution is slightly different from that of \cite{backhoff2022estimating}, to make the paper self-contained, we include proofs of the lemmas in Section~\ref{sec:proof_lemma}.

\begin{remark}[Notation]\label{rmk:notation}
    We introduce the following notation, which is aligned with \cite{backhoff2022estimating}. For notational convenience, in this proof and the proofs of the associated lemmas and corollaries, we write $\mmu^{z_0}_{Y(0)}$ to represent $\mmu^{z_0}_{Y(0), n}$ and $\mmu^{z_1}_{Y(1)}$ to represent $\mmu^{z_1}_{Y(1), m}$. We also write $\hat \pi$ to represent $\hat \pi_{n,m}$. The constants, if not specified explicitly, are independent of $n,m$ but may depend on $d_Z, d_Y$, $L_Z, L_h$.
\end{remark}

Let $\Phi_r^n$ be the small cubes that partition $[0,1]^{d_Z}$ as defined in Definition \ref{defi:adapt_emp}, i.e.,
\begin{align*}
    \Phi_r^n \Let \left\{ (\varphi_r^n)^{-1}(\{x\}): x \in \varphi_r^n([0,1]^{d_Z}) \right\}.
\end{align*}
For the sets $G \subseteq [0,1]^{d_Z}$, we define the conditional distribution of $\mu$ given $G$ as
\[
    \mu^{G}_{Y(0)} = \frac{\int_{G} \mu^{\zeta}_{Y(0)} \diff \mu_{Z}(\zeta)}{\mu_{Z}(G)}.
\]

We also define the conditional distributions of $\hat \mu, \mmu$ given $G$ as
\begin{align*}
    \hat \mu^{G}_{Y(0)} \Let \frac{1}{|I_G|} \sum_{i \in I_G} \delta_{Y_i(0)}, \quad \mmu^{G}_{Y(0)} \Let \frac{1}{|\tilde I_G|} \sum_{i \in \tilde I_G} \delta_{Y_i(0)},
\end{align*}
where $I_G \Let \{i \in [n]: Z_i(0) \in G\}, \tilde I_G \Let \{i \in [n]: \varphi_r^n(Z_i(0)) \in G\}$.

To prove the finite-sample complexity result, we introduce the following lemmas.
\begin{lemma}[Variance of $\Wass_h$]\label{lemm:var_wassh}
    Under Assp.~\ref{a:h_func}, the following holds.
    \[
        \int \Wass_h(\mu_{Y(0)}^{z}, \mu_{Y(1)}^{z})^2 \diff \mu_Z(z) \leq 16 (L_hL_Z)^2 \int \|z\|_2^2 \diff \mu_Z(z) + C_{\mu, h},
    \]
    where $C_{\mu, h}$ is a constant that depends on $\mu, h$.
\end{lemma}

\begin{lemma}[Gap between $\mmu$ and $\hat \mu$]\label{lemm:gap_mmu}
    Under Assp.~\ref{a:compact_domain}, we have
    \begin{align*}
        \Wass_{1}(\mmu_{Z(0)}, \hat \mu_{Z(0)}) \leq \sqrt{d_Z}n^{-r}.
    \end{align*}
\end{lemma}

\begin{lemma}[Gap between $\mmu$ and $\mu$]\label{lemm:gap_mmu2}
Under Assp.~\ref{a:compact_domain} and~\ref{a:Lip_kernel}, we have
    \begin{align*}
        \int \Wass_{1}(\mmu^{z}_{Y(0)}, \mu^{z}_{Y(0)}) \diff \mmu_{Z(0)}(z) 
        \leq & \sum_{G \in \Phi_r^n} \mmu_{Z(0)}(G) \Wass_{1}(\mmu_{Y(0)}^{G}, \mu_{Y(0)}^{G}) + L_Z \sqrt{d_Z}  n^{-r}\\
        =& \sum_{G \in \Phi_r^n} \hat \mu_{Z(0)}(G) \Wass_{1}(\hat \mu_{Y(0)}^{G}, \mu_{Y(0)}^{G}) + L_Z \sqrt{d_Z}  n^{-r}. 
    \end{align*}
\end{lemma}

\begin{lemma}[{\cite[Lemma 3.4]{backhoff2022estimating}}]\label{lemm:rate_emp}
    The following inequalities hold.
    \begin{enumerate}[label=(\roman*)]
        \item Under Assp.~\ref{a:compact_domain} and~\ref{a:sample_dist_random}, $\EE[\Wass_1(\hat \mu_{Z(0)}, \mu_{Z})] \leq C_1 R_{d_{Z}}(n)$.
        \item Under Assp.~\ref{a:compact_domain} and~\ref{a:sample_dist_random},
        \[
            \EE\left[\sum_{G \in \Phi_r^n} \hat \mu_{Z(0)}(G) \Wass_{1}(\hat \mu_{Y(0)}^{G}, \mu_{Y(0)}^{G})\right] \leq C_2R_{d_Y}(n^{1-rd_Z}).
        \]
    \end{enumerate}
    Here, $R_{d}(\cdot)$ is defined in \eqref{eq:R}, $C_1$ is a constant that depends on $d_Z$, and $C_2$ is a constant that depends on $d_Y$.
\end{lemma}

\begin{proof}[{Proof of Theorem~\ref{thm:finite_sample}}]
    We start with
    \begin{align*}
        |\hat V_{n, m} - V_{\cc}| = &  \left|\int \Wass_{h}(\mmu^{z_0}_{Y(0)}, \mmu^{z_1}_{Y(1)}) \diff \hat \pi(z_0, z_1) - V_{\cc}\right|. & (\textup{by Definition~\ref{defi:cot_est}})
    \end{align*}
    For the integrand, we apply Lemma~\ref{lemm:LipWassh} and get
    \begin{align}\label{eq:eq3}
        \left|\Wass_{h}(\mmu^{z_0}_{Y(0)}, \mmu^{z_1}_{Y(1)}) - \Wass_{h}(\mu^{z_0}_{Y(0)}, \mu^{z_1}_{Y(1)})\right| \leq L_h(\Wass_{1}(\mmu^{z_0}_{Y(0)}, \mu^{z_0}_{Y(0)}) + \Wass_{1}( \mmu^{z_1}_{Y(1)}, \mu^{z_1}_{Y(1)})).
    \end{align}
    As a result, we have
    \begin{align}
        & |\hat V_{n, m} - V_{\cc}|  \notag\\
        = & 
        \left|\int \Wass_{h}(\mmu^{z_0}_{Y(0)}, \mmu^{z_1}_{Y(1)}) - \Wass_{h}(\mu^{z_0}_{Y(0)}, \mu^{z_1}_{Y(1)}) \diff \hat \pi(z_0, z_1) + \int \Wass_{h}(\mu^{z_0}_{Y(0)}, \mu^{z_1}_{Y(1)}) \diff \hat \pi(z_0, z_1) - V_{\cc}\right|\notag\\
        \leq & \int \left|\Wass_{h}(\mmu^{z_0}_{Y(0)}, \mmu^{z_1}_{Y(1)}) - \Wass_{h}(\mu^{z_0}_{Y(0)}, \mu^{z_1}_{Y(1)})\right| \diff \hat \pi(z_0, z_1) + \left| \int \Wass_{h}(\mu^{z_0}_{Y(0)}, \mu^{z_1}_{Y(1)}) \diff \hat \pi(z_0, z_1) - V_{\cc}\right| \notag\\
        \leq 
        & L_h \underbrace{\int \Wass_{1}(\mmu^{z_0}_{Y(0)}, \mu^{z_0}_{Y(0)}) \diff \mmu_{Z(0)}(z_0) + \int \Wass_{1}(\mu^{z_1}_{Y(1)}, \mmu^{z_1}_{Y(1)}) \diff \mmu_{Z(1)}(z_1)}_{(\textup{Term A})} \qquad\ (\textup{by Equation~\eqref{eq:eq3}})\notag\\
        & + \underbrace{\left|\int \Wass_{h}(\mu^{z_0}_{Y(0)}, \mu^{z_1}_{Y(1)}) \diff \hat \pi(z_0, z_1) - V_{\cc}\right|}_{(\textup{Term B})}.\label{eq:eq494}
    \end{align}

    For \textbf{(Term B)}, again, we apply Lemma~\ref{lemm:LipWassh} to the integrand and get
    \begin{align*}
        \left|\Wass_{h}(\mu^{z_0}_{Y(0)}, \mu^{z_1}_{Y(1)}) - \Wass_{h}(\mu^{z_0}_{Y(0)}, \mu^{z_0}_{Y(1)})\right| \leq & L_h \Wass_1(\mu^{z_0}_{Y(1)}, \mu^{z_1}_{Y(1)}) \leq  L_h L_Z \lnorm{z_0 - z_1}{2},
    \end{align*}
    where the second inequality is due to Assp.~\ref{a:Lip_kernel}. Therefore, we have
    \begin{align}
        & (\textup{Term B}) \notag\\
        \leq & \left|\int \Wass_{h}(\mu^{z_0}_{Y(0)}, \mu^{z_1}_{Y(1)}) - \Wass_{h}(\mu^{z_0}_{Y(0)}, \mu^{z_0}_{Y(1)}) \diff \hat \pi(z_0, z_1)\right| + \left| \int \Wass_{h}(\mu^{z_0}_{Y(0)}, \mu^{z_0}_{Y(1)}) \diff \hat \pi(z_0, z_1) - V_{\cc}\right|\notag\\
        \leq &\underbrace{L_h L_Z \int \lnorm{z_0 - z_1}{2} \diff \hat \pi(z_0, z_1)}_{\textup{(Term C)}} + \underbrace{\left| \int \Wass_{h}(\mu^{z_0}_{Y(0)}, \mu^{z_0}_{Y(1)}) \diff \mmu_{Z(0)}(z_0) - V_{\cc}\right|}_{\textup{(Term D)}}. \label{eq:518}
    \end{align}

    For \textbf{(Term D)}, we have
    \begin{align*}
        & \EE\left[\left| \int \Wass_{h}(\mu^{z_0}_{Y(0)}, \mu^{z_0}_{Y(1)}) \diff \mmu_{Z(0)}(z_0) - V_{\cc}\right|\right]\\
        \leq & \EE\left[\left| \int \Wass_{h}(\mu^{z_0}_{Y(0)}, \mu^{z_0}_{Y(1)}) (\diff \mmu_{Z(0)}(z_0) - \diff \hat \mu_{Z(0)}(z_0))\right|\right] + \EE\left[\left| \int \Wass_{h}(\mu^{z_0}_{Y(0)}, \mu^{z_0}_{Y(1)}) \diff \hat \mu_{Z(0)}(z_0) - V_{\cc}\right|\right]. 
    \end{align*}
    By Lemma~\ref{lemm:LipWassh} and Assp.~\ref{a:Lip_kernel}, the mapping $z \mapsto \Wass_{h}(\mu^{z}_{Y(0)}, \mu^{z}_{Y(1)})$ is $(2L_h L_Z)$-Lipschitz continuous. Then, by Lemma~\ref{lemm:gap_mmu}, 
    \begin{align*}
        \EE\left[\left| \int \Wass_{h}(\mu^{z_0}_{Y(0)}, \mu^{z_0}_{Y(1)}) (\diff \mmu_{Z(0)}(z_0) - \diff \hat \mu_{Z(0)}(z_0))\right|\right]
        \leq 2 L_h L_Z \Wass_1(\mmu_{Z(0)}, \hat \mu_{Z(0)}) \leq 2L_h L_Z \sqrt{d_Z} n^{-r}.
    \end{align*}
    By Assp.~\ref{a:sample_dist_random}, Lemma~\ref{lemm:var_wassh} and Markov's inequality, 
    \begin{align*}
        \EE\left[\left| \int \Wass_{h}(\mu^{z_0}_{Y(0)}, \mu^{z_0}_{Y(1)}) \diff \hat \mu_{Z(0)}(z_0) - V_{\cc}\right|\right]
        \leq & \left(16 (L_hL_Z)^2 \int \|z\|_2^2 \diff \mu_Z(z) + C_{\mu, h}\right)^{\frac{1}{2}} n^{-\frac{1}{2}}\\
        \leq & \hat C n^{-r},
    \end{align*}
    where $r = \frac{1}{d_Z + 2 \vee d_Y} \leq \frac{1}{3}$.
    
    Note that, the $n^{-\frac{1}{2}}$ above can be replaced by $(m \vee n)^{-\frac{1}{2}}$ since $Z(0), Z(1)$ are symmetric thus can be swapped in the derivation above.

    For \textbf{(Term C)}, we apply Assp.~\ref{a:coupling_gap}.
    Finally, combine (Term C) and (Term D), and we arrive at
    \begin{align*}
        \EE[\textup{(Term B)}] \leq & C' (m \wedge n)^{-r}.
    \end{align*}
    
    For \textbf{(Term A)}, we present the upper bound for $\int \Wass_{1}(\mmu^{z_0}_{Y(0)}, \mu^{z_0}_{Y(0)}) \diff \mmu_{Z(0)}(z_0)$ and the remaining part can be derived similarly.

    By Lemma~\ref{lemm:gap_mmu2}, 
    \begin{align*}
        \int \Wass_{1}(\mmu^{z}_{Y(0)}, \mu^{z}_{Y(0)}) \diff \mmu_{Z(0)}(z) \leq \sum_{G \in \Phi_r^n} \hat \mu_{Z(0)}(G) \Wass_{1}(\hat \mu_{Y(0)}^{G}, \mu_{Y(0)}^{G}) + L_Z \sqrt{d_Z}  n^{-r}.
    \end{align*}
    Then, apply Lemma~\ref{lemm:rate_emp}, we get
    \begin{align*}
        \EE\left[\int \Wass_{1}(\mmu^{z}_{Y(0)}, \mu^{z}_{Y(0)}) \diff \mmu_{Z(0)}(z)\right]
        \leq & \EE\left[\sum_{G \in \Phi_r^n} \hat \mu_{Z(0)}(G) \Wass_{1}(\hat \mu_{Y(0)}^{G}, \mu_{Y(0)}^{G})\right] + L_Z \sqrt{d_Z}  n^{-r} \\
        \leq & C_2 R_{d_Y}(n^{1-rd_Z}) + L_Z \sqrt{d_Z}  n^{-r}.
    \end{align*}
    $R_{d_Y}(\cdot)$ is defined in \eqref{eq:R}, and $C_2$ is a constant that depends on $d_Y$ and $\max_{y \in \calY} \lnorm{y}{2}$.

    Therefore, we get
    \begin{align*}
        \EE[(\textup{Term A})] \leq \tilde C(R_{d_Y}((m \wedge n)^{1-rd_Z}) + (m \wedge n)^{-r}).
    \end{align*}
    
    Finally, combine (Term A) and (Term B), we get
    \begin{align*}
        \EE\left[|\hat V_{n,m} - V_{\cc}|\right] \leq \bar C(R_{d_Y}((m \wedge n)^{1-rd_Z}) + (m \wedge n)^{-r}).
    \end{align*}

    Let $r = \frac{1}{d_Z + 2 \vee d_Y}$, which minimizes the right-hand side on the order of $m \wedge n$, we get
    \[
        \EE\left[|\hat V_{n,m} - V_{\cc}|\right] \leq 
    \begin{cases}
        \bar C (m \wedge n)^{-\frac{1}{d_Z + 2 \vee d_Y}} \log(m \wedge n) & \textup{if}\  d_Y \neq 2,\\
        \bar C (m \wedge n)^{-\frac{1}{d_Z + 2 \vee d_Y}} & \textup{if}\  d_Y = 2.
    \end{cases}
    \]
\end{proof}

\begin{remark}[Choice of $\hat \pi_{n,m}$ in \eqref{eq:pi_choice_ot}]\label{rmk:G3} When $\hat \pi_{n,m}$ is the optimal coupling of the marginals, by Lemma~\ref{lemm:gap_mmu} and the triangle inequality (Lemma~\ref{lem:tri_ineq}), 
    \begin{align*}
        & \int \lnorm{z_0 - z_1}{2} \diff \hat \pi(z_0, z_1)\\
        = & \EE\left[\Wass_1(\mmu_{Z(0)}, \mmu_{Z(1)})\right] \\
        \leq & \EE\left[\Wass_1(\mmu_{Z(0)}, \hat \mu_{Z(0)})\right] + \EE\left[\Wass_1(\hat \mu_{Z(0)}, \mu_{Z})\right] + \EE\left[\Wass_1(\hat \mu_{Z(1)}, \mu_{Z})\right] + \EE\left[\Wass_1(\hat \mu_{Z(1)}, \mmu_{Z(1)})\right]\\
        \leq & 2 \sqrt{d_Z}(n^{-r} + m^{-r}) + \EE\left[\Wass_1(\hat \mu_{Z(0)}, \mu_{Z})\right] +  \EE\left[\Wass_1(\hat \mu_{Z(1)}, \mu_{Z})\right].
    \end{align*}

    Then, by Lemma~\ref{lemm:rate_emp} (i), we get
    \[
        \int \lnorm{z_0 - z_1}{2} \diff \hat \pi(z_0, z_1) \leq 2 \sqrt{d_Z}(n^{-r} + m^{-r}) + 2 C_1 R_{d_Z}(n \wedge m).
    \]
    Then, for $r = \frac{1}{d_Z + 2 \vee d_Y}$, we have 
    \[
        R_{d_Z}(n \wedge m) \leq C'(n \wedge m)^{-\frac{1}{2 \vee d_Z + 1}} \leq C'(n \wedge m)^{-r}.
    \]
    Therefore, $\hat \pi_{n,m}$ in \eqref{eq:pi_choice_ot} satisfies Assp.~\ref{a:coupling_gap} when $r = \frac{1}{d_Z + 2 \vee d_Y}$.
\end{remark}

\subsection{Proof of Theorem~\ref{rmk:robust}}
\begin{proof}[Proof of Theorem~\ref{rmk:robust}]
    If the generating distribution of $Z(0), Z(1)$ are not equal to $\mu_Z$, i.e., $\mu_{Z(0)} \neq \mu_{Z}$ and $\mu_{Z(1)} \neq \mu_{Z}$, then for $r = \frac{1}{d_Z + 2 \vee d_Y}$, according to the derivation in Remark~\ref{rmk:G3}, we have
    \[
        \begin{aligned}
            & \int \lnorm{z_0 - z_1}{2} \diff \hat \pi(z_0, z_1)\\
            \leq & 2 \sqrt{d_Z}(n^{-r} + m^{-r}) + \EE\left[\Wass_1(\hat \mu_{Z(0)}, \mu_{Z})\right] +  \EE\left[\Wass_1(\hat \mu_{Z(1)}, \mu_{Z})\right]\\
            \leq & 2 \sqrt{d_Z}(n^{-r} + m^{-r}) + \EE\left[\Wass_1(\hat \mu_{Z(0)}, \mu_{Z(0)})\right] +  \EE\left[\Wass_1(\hat \mu_{Z(1)}, \mu_{Z(1)})\right] + \Wass_1(\mu_{Z(0)}, \mu_{Z}) + \Wass_1(\mu_{Z(1)}, \mu_{Z})\\
            \leq & C'(n \wedge m)^{-r} + 2\epsilon.
        \end{aligned}
    \]

    Plug this condition back to the derivation in the proof of Theorem~\ref{thm:finite_sample}, specifically \eqref{eq:518}, we get the desired result.
\end{proof}

\subsection{Proof of Theorem~\ref{thm:consist}}
\begin{lemma}[Convergence of adapted empirical distribution]\label{lemm:conv_emp_dist}
    Under Assp.~\ref{a:compact_domain} - \ref{a:Lip_kernel}, then almost surely, 
    \[
    \lim_{n \rightarrow \infty} \Wass_{\bc}(\mmu_{Y(0), Z(0), n}, \mu_{Y(0), Z}) = \lim_{m \rightarrow \infty} \Wass_{\bc}(\mmu_{Y(1), Z(1), m}, \mu_{Y(1), Z}) = 0.
    \]
\end{lemma}
\begin{proof}[Proof of Theorem~\ref{thm:consist}]
    The proof strategy of Theorem~\ref{thm:consist} is similar to that of Theorem~\ref{thm:finite_sample}. By \eqref{eq:eq494}, we have
    \[
        \begin{aligned}
        & |\hat V_{n,m} - V_{\cc}| \\
        \leq 
        & \underbrace{\int \Wass_{1}(\mmu^{z_0}_{Y(0), n}, \mu^{z_0}_{Y(0)}) \diff \mmu_{Z(0), n}(z_0) + \int \Wass_{1}(\mu^{z_1}_{Y(1)}, \mmu^{z_1}_{Y(1), m}) \diff \mmu_{Z(1), m}(z_1)}_{(\textup{Term A})}\\
        & + \underbrace{\left|\int \Wass_{h}(\mu^{z_0}_{Y(0)}, \mu^{z_1}_{Y(1)}) \diff \hat \pi_{n, m}(z_0, z_1) - V_{\cc}\right|}_{(\textup{Term B})}.
    \end{aligned}
    \]

    For \textbf{(Term A)}, using a similar argument for Eq.~\eqref{eq:marginal_to_couple}, we get
    \[
        \text{(Term A)} \leq (L_Z \vee 1) (\Wass_{\bc}(\mmu_{Y(0), Z(0), n}, \mu_{Y(0), Z}) + \Wass_{\bc}(\mmu_{Y(1), Z(1), m}, \mu_{Y(1), Z})).
    \]
    By Lemma~\ref{lemm:conv_emp_dist}, we get (Term A) converges to zero almost surely.

    For \textbf{(Term B)}, by Assp.~\ref{a:compact_domain},~\ref{a:Lip_kernel}, $(z_0, z_1) \mapsto \Wass_{h}(\mu^{z_0}_{Y(0)}, \mu^{z_1}_{Y(1)})$ is a bounded continuous function. Indeed, by Assp.~\ref{a:Lip_kernel} and Lemma~\ref{lemm:LipWassh},
    \[
        \begin{aligned}
             \Wass_{h}(\mu^{z_0}_{Y(0)}, \mu^{z_1}_{Y(1)}) -  \Wass_{h}(\mu^{z_0'}_{Y(0)}, \mu^{z_1'}_{Y(1)})
             \leq & L_h (\Wass_1(\mu^{z_0}_{Y(0)}, \mu^{z_0'}_{Y(0)}) + \Wass_1(\mu^{z_1}_{Y(1)}, \mu^{z_1'}_{Y(1)})) \\
             \leq & L_h (\lnorm{z_0 - z_0'}{2} + \lnorm{z_1 - z_1'}{2}).
        \end{aligned}
    \]
    Further, by Assp.~\ref{a:compact_domain}, we get that $\Wass_{h}(\mu^{z_0}_{Y(0)}, \mu^{z_1}_{Y(1)})$ is bounded on $\calZ$.

    Then, since $\hat \pi_{n, m}$ weakly converges to $(\mathrm{id}, \mathrm{id})_{\#}\mu_Z$ a.s., we get 
    \[
        \begin{aligned}
            &\lim_{n,m\rightarrow\infty} \int \Wass_{h}(\mu^{z_0}_{Y(0)}, \mu^{z_1}_{Y(1)}) \diff \hat \pi_{n, m}(z_0, z_1) - V_{\cc}\\
            =& \lim_{n,m\rightarrow\infty} \int \Wass_{h}(\mu^{z_0}_{Y(0)}, \mu^{z_1}_{Y(1)}) \diff \hat \pi_{n, m}(z_0, z_1) - \int \Wass_{h}(\mu^{z}_{Y(0)}, \mu^{z}_{Y(1)}) \diff \mu_Z(z) = 0.
        \end{aligned}
    \]
    Thus, (Term B) converges to zero almost surely. The desired result is proved.
\end{proof}

\subsection{Proof of Theorem~\ref{thm:finite_sampleII}}
\begin{lemma}[TV bound]\label{lem:TV_bound}
    Suppose $Z, Z' \in \calZ$, and $r = \sup_{z, z' \in \calZ} \lnorm{z - z'}{2} < \infty$. Then,
    \[
        \Wass_{1}(Z,Z') \leq r \times d_{\text{TV}}(Z,Z') = \frac{r}{2}\int |\diff \mu_Z - \diff \mu_{Z'}|.
    \]
\end{lemma}

\begin{lemma}[Linear functional gap]\label{lem:linear_gap}
    Suppose that $(Z_i, 1\leq i \leq n)$ are i.i.d.~drawn from $\mu \in \calP([0,1]^{d_Z})$, and let $\hat \mu_{Z(0), n} = \sum_{i=1}^n \hat w_i \delta_{Z_i}$, where $\hat w_i$ is defined as
    \begin{align*}
    \hat w_i = \frac{(1-\hat e(Z_i(0)))^{-1}}{\sum_{l=1}^n (1-\hat e(Z_l(0)))^{-1}},\  i \in [n].
    \end{align*}

    Additionally, let $\tilde \mu_{Z(0), n} = \sum_{i=1}^{n} w_i \delta_{Z_i(0)}$, where $w_i$ is defined as
    \[
        \begin{aligned}    
        w_i = \frac{\frac{\diff \mu}{\diff \mu_{|W = 0}}(Z_i(0))}{\sum_{i=1}^{n} \frac{\diff \mu}{\diff \mu_{|W = 0}}(Z_i(0))}\quad i\in [n].
        \end{aligned}
    \]
    
    Under Assp.~\ref{assu:overlap} and~\ref{a:est_propensity_score}, we have
    \[
        \EE[\Wass_1(\hat \mu_{Z(0), n}, \tilde \mu_{Z(0), n})] \leq \frac{\sqrt{d_Z}}{2\delta^3}n^{-\frac{1}{2}} + \frac{\sqrt{d_Z} C_{\textup{w}}}{2 \delta} n^{-r}.
    \]
\end{lemma}

\begin{lemma}[Empirical Wasserstein convergence rate for reweighted empirical distribution]\label{lem:Wass_conv_reweight}
    Under the same condition of Lemma~\ref{lem:linear_gap}, we have 
    \[
    \EE[\Wass_1(\hat \mu_{Z(0), n}, \mu_{Z(0), n})] \leq \delta^{-2} C R_{d_Z}(n) + \frac{\sqrt{d_Z}}{2\delta^3}n^{-\frac{1}{2}} + \frac{\sqrt{d_Z} C_{\textup{w}}}{2 \delta} n^{-r},\]
    where $C$ is a constant depending on $d_Z$.
\end{lemma}

\begin{lemma}[Discretization gap]\label{lem:discrete_gapII}
    Under Assp.~\ref{a:lip_propensity_score} and~\ref{a:bound_e} with
    \[
        \begin{aligned}
        \mmu_{Z(0), n} = \sum_{G \in \Phi_r^n} \frac{(1-\hat e(\varphi_r^n(G)))^{-1} |\{i: Z_i(0) \in G\}|}{\sum_{G} (1-\hat e(\varphi_r^n(G)))^{-1} |\{i: Z_i(0) \in G\}|}\delta_{\varphi_r^n(Z_i(0))},   
        \end{aligned}
    \]
    and 
    \[
        \hat \mu_{Z(0), n} = \sum_{i=1}^{n} \hat w_i \delta_{Z_i(0)},
    \]
    where $\hat w_i$ is defined as in Lemma~\ref{lem:linear_gap}, we have
    \[
        \begin{aligned}
            \Wass_1(\mmu_{Z(0), n}, \hat \mu_{Z(0), n}) \leq \sqrt{d_Z} \left(1+\frac{L_e}{\eta^3}\right) n^{-r}.
        \end{aligned}
    \]
\end{lemma}

\begin{lemma}[Gap between $\mmu$ and $\mu$]\label{lemm:gap_mmuIII}
Under Assp.~\ref{a:compact_domain} and~\ref{a:Lip_kernel} with $\mmu_{Z(0), n}$ defined as in Lemma~\ref{lem:discrete_gapII}, we have
    \begin{align*}
        \int \Wass_{1}(\mmu^{z}_{Y(0), n}, \mu^{z}_{Y(0)}) \diff \mmu_{Z(0), n}(z) 
        \leq & \sum_{G \in \Phi_r^n} \mmu_{Z(0), n}(G) \Wass_{1}(\mmu_{Y(0), n}^{G}, \mu_{Y(0)}^{G}) + L_Z \sqrt{d_Z}  n^{-r}.
    \end{align*}
\end{lemma}
\begin{proof}[Proof of Lemma~\ref{lemm:gap_mmuIII}]
    The proof follows exactly the same as the proof of Lemma~\ref{lemm:gap_mmu2}, which holds as long as $\mmu_{Z(0), n}$ is a probability measure supported on the cell centers $\{\varphi_r^n(G): G \in \Phi_r^n\}$.
\end{proof}

We also introduce the following theorem adapted from \cite[Lemma 3.3]{backhoff2022estimating}.
\begin{lemma}[Conditional distribution of $Y$]\label{lem:conditional_independ_y}
    Under Assp.~\ref{a:covariate_shift}, conditionally on $\{|\{i: Z_i(0) \in G\}|, G \in \Phi_r^n\}$, for each $G$, the sample $(Y_i(0), i: Z_i(0) \in G)$ i.i.d.~follows $\mu(\cdot | Z(0) \in G, W=0)$.
\end{lemma}

\begin{proof}[Proof of Theorem~\ref{thm:finite_sampleII}]
    Without loss of generality, due to the cross-fitting, we can assume that the random estimation of propensity score function $\hat e$ is independent of $((Y_i(0), Z_i(0)), i \in [n])$ and $((Y_j(1), Z_j(1)), j \in [m])$.

    The proof strategy of Theorem~\ref{thm:consist} is similar to that of Theorem~\ref{thm:finite_sample}. 

    Recall that
    \[
    \begin{aligned}
        \mmu_{Y(0), Z(0), n} &= \sum_{G \in \Phi_r^n} \frac{(1-\hat e(\varphi_r^n(G)))^{-1}}{\sum_{G} (1-\hat e(\varphi_r^n(G)))^{-1} |\{i: Z_i(0) \in G\}|}\sum_{i: Z_i(0) \in G}\delta_{ Y_i(0), \varphi_r^n(Z_i(0))},\\
        \mmu_{Y(1), Z(1), m} &= \sum_{G \in \Phi_r^m} \frac{\hat e(\varphi_r^m(G))^{-1}}{\sum_{G} \hat e(\varphi_r^m(G))^{-1} |\{j: Z_j(1) \in G\}|} \sum_{j: Z_j(1) \in G}  \delta_{Y_j(1), \varphi_r^m(Z_j(0))},
    \end{aligned}
    \]

    First, note that 
    \[
    \begin{aligned}
        \mmu_{Z(0), n} &= \sum_{G \in \Phi_r^n} \frac{(1-\hat e(\varphi_r^n(G)))^{-1} |\{i: Z_i(0) \in G\}|}{\sum_{G} (1-\hat e(\varphi_r^n(G)))^{-1} |\{i: Z_i(0) \in G\}|}\delta_{\varphi_r^n(Z_i(0))},\\
        \mmu_{Z(1), m} &= \sum_{G \in \Phi_r^m} \frac{\hat e(\varphi_r^m(G))^{-1}|\{j: Z_j(1) \in G\}|}{\sum_{G} \hat e(\varphi_r^m(G))^{-1} |\{j: Z_j(1) \in G\}|} \delta_{ \varphi_r^m(Z_j(0))}.
    \end{aligned}
    \]
    and $\hat \pi_{n,m}$ is an optimal coupling of $\mmu_{Z(0), n}$ and $\mmu_{Z(1), m}$ associated with $\Wass_1(\mmu_{Z(0), n}, \mmu_{Z(1), m})$.
    For the conditional distributions, we have, for $G \in \Phi_r^n, G' \in \Phi_r^m$
    \[
    \begin{aligned}
        \mmu_{Y(0), n}^G &= \frac{1}{|\{i: Z_i(0) \in G\}|} \sum_{i: Z_i(0) \in G} \delta_{Y_i(0)},\\
        \mmu_{Y(1), m}^{G'} &= \frac{1}{|\{j: Z_j(1) \in G'\}|} \sum_{j: Z_j(1) \in G'} \delta_{Y_j(1)}.
    \end{aligned}
    \]
    Additionally, we let
    \[
    \hat \mu_{Z(0), n} = \sum_{i=1}^{n} \hat w_i \delta_{Z_i(0)},\quad \hat \mu_{Z(1), m} =  \sum_{j=1}^{m} \hat \xi_j \delta_{Z_j(1)}.
    \]
    and
    \[
        \tilde \mu_{Z(0), n} = \sum_{i=1}^{n} w_i \delta_{Z_i(0)},\quad \tilde \mu_{Z(1), m} =  \sum_{j=1}^{m} \xi_j \delta_{Z_j(1)},
    \]
    where for $i \in [n], j \in [m]$,
    \[
        w_i = \frac{\frac{\diff \mu}{\diff \mu_{|W = 0}}(Z_i(0))}{\sum_{i=1}^{n} \frac{\diff \mu}{\diff \mu_{|W = 0}}(Z_i(0))},\quad \xi_j = \frac{\frac{\diff \mu}{\diff \mu_{|W=1}}(Z_j(1))}{\sum_{j=1}^{m} \frac{\diff \mu}{\diff \mu_{|W=1}}(Z_j(1))}.
    \]

    In the following, we adopt the same notation as in the proof of Theorem~\ref{thm:finite_sample}, omitting the dependence of $n,m$ in subscripts and superscripts for notational simplicity.
    
    By \eqref{eq:eq494}, we have
    \[
        \begin{aligned}
         |\hat V_{n,m} - V_{\cc}|
        \leq 
        & \underbrace{L_h \left(\int \Wass_{1}(\mmu^{z_0}_{Y(0)}, \mu^{z_0}_{Y(0)}) \diff \mmu_{Z(0)}(z_0) + \int \Wass_{1}(\mu^{z_1}_{Y(1)}, \mmu^{z_1}_{Y(1)}) \diff \mmu_{Z(1)}(z_1)\right)}_{(\textup{Term A})}\\
        & + \underbrace{\left|\int \Wass_{h}(\mu^{z_0}_{Y(0)}, \mu^{z_1}_{Y(1)}) \diff \hat \pi(z_0, z_1) - V_{\cc}\right|}_{(\textup{Term B})}.
    \end{aligned}
    \]

    For \textbf{(Term B)}, by Eq.\eqref{eq:518}, we have
    \[
        \begin{aligned}
        & (\textup{Term B}) \notag\\
        \leq & \left|\int \Wass_{h}(\mu^{z_0}_{Y(0)}, \mu^{z_1}_{Y(1)}) - \Wass_{h}(\mu^{z_0}_{Y(0)}, \mu^{z_0}_{Y(1)}) \diff \hat \pi(z_0, z_1)\right| + \left| \int \Wass_{h}(\mu^{z_0}_{Y(0)}, \mu^{z_0}_{Y(1)}) \diff \hat \pi(z_0, z_1) - V_{\cc}\right|\notag\\
        \leq &\underbrace{L_h L_Z \int \lnorm{z_0 - z_1}{2} \diff \hat \pi(z_0, z_1)}_{\textup{(Term C)}} + \underbrace{\left| \int \Wass_{h}(\mu^{z_0}_{Y(0)}, \mu^{z_0}_{Y(1)}) \diff \mmu_{Z(0)}(z_0) - V_{\cc}\right|}_{\textup{(Term D)}}. 
    \end{aligned}
    \]
    
    For \textbf{(Term C)}, similar to the derivation in Remark~\ref{rmk:G3}, we have
    \begin{align*}
        & \int \lnorm{z_0 - z_1}{2} \diff \hat \pi(z_0, z_1)\\
        = & \EE\left[\Wass_1(\mmu_{Z(0)}, \mmu_{Z(1)})\right] \\
        \leq & \EE\left[\Wass_1(\mmu_{Z(0)}, \hat \mu_{Z(0)})\right] + \EE\left[\Wass_1(\hat \mu_{Z(0)}, \mu_{Z})\right] + \EE\left[\Wass_1(\hat \mu_{Z(1)}, \mu_{Z})\right] + \EE\left[\Wass_1(\hat \mu_{Z(1)}, \mmu_{Z(1)})\right]\\
        \leq & 2 \sqrt{d_Z} (1+\frac{L_e}{\eta^3})(n^{-r} + m^{-r}) + \EE\left[\Wass_1(\hat \mu_{Z(0)}, \mu_{Z})\right] +  \EE\left[\Wass_1(\hat \mu_{Z(1)}, \mu_{Z})\right],
    \end{align*}
    where the first inequality is the triangle inequality (Lemma~\ref{lem:tri_ineq}) and the second is due to Lemma~\ref{lem:discrete_gapII}.

    Under Assp.~\ref{assu:overlap} and ~\ref{a:est_propensity_score}, by Lemma~\ref{lem:Wass_conv_reweight}, we get
    \[
        \EE\left[\Wass_1(\hat \mu_{Z(0)}, \mu_{Z})\right] \leq \delta^{-2} C R_{d_Z}(n) + \sqrt{d_Z} C_{\textup{w}} n^{-r}.
    \]
    This is due to $\frac{\diff \mu}{\diff \mu_{|W = 0}}(z) = (1-e(z))^{-1} \PP(W=0) \leq \delta^{-1}$ (Assp.~\ref{a:est_propensity_score}).
    Combine these, we get
    \[
        \text{(Term C)} \leq C'_1(R_{d_Z}(n\wedge m) + (n\wedge m)^{-r}).
    \]
    
    For \textbf{(Term D)}, we have
    \begin{align*}
        & \EE\left[\left| \int \Wass_{h}(\mu^{z_0}_{Y(0)}, \mu^{z_0}_{Y(1)}) \diff \mmu_{Z(0)}(z_0) - V_{\cc}\right|\right]\\
        \leq & \EE\left[\left| \int \Wass_{h}(\mu^{z_0}_{Y(0)}, \mu^{z_0}_{Y(1)}) (\diff \mmu_{Z(0)}(z_0) - \diff \hat \mu_{Z(0)}(z_0))\right|\right] + \EE\left[\left| \int \Wass_{h}(\mu^{z_0}_{Y(0)}, \mu^{z_0}_{Y(1)}) \diff \hat \mu_{Z(0)}(z_0) - V_{\cc}\right|\right]. 
    \end{align*}
    By Lemma~\ref{lemm:LipWassh} and Assp.~\ref{a:Lip_kernel}, the mapping $z \mapsto \Wass_{h}(\mu^{z}_{Y(0)}, \mu^{z}_{Y(1)})$ is $(2L_h L_Z)$-Lipschitz continuous. Then, by Lemma~\ref{lem:discrete_gapII}, 
    \begin{align*}
        \EE\left[\left| \int \Wass_{h}(\mu^{z_0}_{Y(0)}, \mu^{z_0}_{Y(1)}) (\diff \mmu_{Z(0)}(z_0) - \diff \hat \mu_{Z(0)}(z_0))\right|\right]
        \leq & 2 L_h L_Z \Wass_1(\mmu_{Z(0)}, \hat \mu_{Z(0)})\\
        \leq & 2L_h L_Z \sqrt{d_Z}(1+\frac{L_e}{\eta^3}) n^{-r}.
    \end{align*}
    On the other hand,
    \[
        \begin{aligned}
            &\EE\left[\left| \int \Wass_{h}(\mu^{z_0}_{Y(0)}, \mu^{z_0}_{Y(1)}) \diff \hat \mu_{Z(0)}(z_0) - V_{\cc}\right|\right]\\
            \leq & \underbrace{\EE\left[\left| \int \Wass_{h}(\mu^{z_0}_{Y(0)}, \mu^{z_0}_{Y(1)}) \diff \hat \mu_{Z(0)}(z_0) - \int \Wass_{h}(\mu^{z_0}_{Y(0)}, \mu^{z_0}_{Y(1)}) \diff \tilde \mu_{Z(0)}(z_0)\right|\right]}_{(\text{Term E})}\\
            & + \underbrace{\EE\left[\left| \int \Wass_{h}(\mu^{z_0}_{Y(0)}, \mu^{z_0}_{Y(1)}) \diff \tilde \mu_{Z(0)}(z_0) - V_{\cc}\right|\right]}_{(\text{Term F})}
        \end{aligned}
    \]
    By Lemma~\ref{lem:linear_gap}, and the fact that the mapping $z \mapsto \Wass_{h}(\mu^{z}_{Y(0)}, \mu^{z}_{Y(1)})$ is $(2L_h L_Z)$-Lipschitz continuous, we get
    \[
        (\text{Term E}) \leq 2 L_h L_Z (\frac{\sqrt{d_Z}}{2\delta^3}n^{-\frac{1}{2}} + \frac{\sqrt{d_Z} C_{\textup{w}}}{2 \delta} n^{-r}).
    \]    
    And by Assp.~\ref{assu:unconfoundedness}, Lemma~\ref{lemm:var_wassh} and Markov's inequality, we get
    \[
        (\text{Term F}) \leq \delta^{-1}\left(16 (L_hL_Z)^2 \int \|z\|_2^2 \diff \mu_Z(z) + C_{\mu, h}\right)^{\frac{1}{2}} n^{-\frac{1}{2}} \leq C'_2 n^{-r},
    \]
    where $r = \frac{1}{d_Z + 2 \vee d_Y} \leq \frac{1}{3}$.    
    Combine these, we get
    \[
        \text{(Term D)} \leq C'_3(n\wedge m)^{-r}.
    \]

    For \textbf{(Term A)}, we present the upper bound for $\int \Wass_{1}(\mmu^{z_0}_{Y(0)}, \mu^{z_0}_{Y(0)}) \diff \mmu_{Z(0)}(z_0)$ and the remaining part can be derived similarly.

    Let $\nu = \mu_{Z | W = 0} \otimes \mu_{Y(0)}^{Z}$, thus by Assp.~\ref{assu:unconfoundedness}, $\nu^{z}_{Y(0)} = \mu^{z}_{Y(0)}$. Then, by Lemma~\ref{lemm:gap_mmu2},
    \begin{align*}
        &\int \Wass_{1}(\mmu^{z}_{Y(0)}, \mu^{z}_{Y(0)}) \diff \mmu_{Z(0)}(z)\\
        =& \int \Wass_{1}(\mmu^{z}_{Y(0)}, \nu^{z}_{Y(0)}) \diff \mmu_{Z(0)}(z)\\
        \leq & \sum_{G \in \Phi_r^n} \mmu_{Z(0)}(G) \Wass_{1}(\mmu_{Y(0)}^{G}, \nu_{Y(0)}^{G}) + L_Z \sqrt{d_Z}  n^{-r}\\
        =& \sum_{G \in \Phi_r^n} \mmu_{Z(0)}(G) \Wass_{1}\left(\frac{1}{|\{i: Z_i(0) \in G\}|} \sum_{i: Z_i(0) \in G} \delta_{Y_i(0)}, \nu_{Y(0)}^{G}\right) + L_Z \sqrt{d_Z}  n^{-r}.
    \end{align*}

    Note that $(\mmu_{Z(0)}(G), G \in \Phi_r^n)$ only depends on $\hat e$ and $\calG = \{|\{i: Z_i(0) \in G\}|, G \in \Phi_r^n\}$, then apply Lemma~\ref{lem:conditional_independ_y}, we get,
    \[
    \begin{aligned}
        &\EE\left[\sum_{G \in \Phi_r^n} \mmu_{Z(0)}(G) \Wass_{1}\left(\frac{1}{|\{i: Z_i(0) \in G\}|} \sum_{i: Z_i(0) \in G} \delta_{Y_i(0)}, \nu_{Y(0)}^{G}\right)\right]\\
        = & \EE\left[\sum_{G \in \Phi_r^n} \mmu_{Z(0)}(G) \EE\left[\Wass_{1}\left(\frac{1}{|\{i: Z_i(0) \in G\}|} \sum_{i: Z_i(0) \in G} \delta_{Y_i(0)}, \nu_{Y(0)}^{G}\right) | \calG, \hat e \right]\right]\\
        \leq & C'_4 \EE\left[\sum_{G \in \Phi_r^n} \frac{(1-\hat e(\varphi_r^n(G)))^{-1} |\{i: Z_i(0) \in G\}|}{\sum_{G} (1-\hat e(\varphi_r^n(G)))^{-1} |\{i: Z_i(0) \in G\}|} R_{d_Y}(|\{i: Z_i(0) \in G\}|)\right]\\
        \leq & \frac{C'_4 \eta^{-1}}{n} \EE\left[\sum_{G \in \Phi_r^n} |\{i: Z_i(0) \in G\}| R_{d_Y}(|\{i: Z_i(0) \in G\}|)\right]\\
        = & \frac{C'_4 \eta^{-1}|\Phi_r^n|}{n} \EE\left[\frac{1}{|\Phi_r^n|}\sum_{G \in \Phi_r^n} |\{i: Z_i(0) \in G\}| R_{d_Y}(|\{i: Z_i(0) \in G\}|)\right]\\
        \leq& C'_4 \eta^{-1} R_{d_Y}(n^{1-r d_Z}),
    \end{aligned}
    \]
    where the first inequality is also due to Lemma~\ref{lem:conditional_independ_y} with independence between $\hat e$ and the sample $Z, Y$, the second is due to Assp.~\ref{a:bound_e}, and the last inequality is due to $x \mapsto x R_{d_Y}(x)$ is concave.

    As a result,
    \begin{align*}
        \EE\left[\int \Wass_{1}(\mmu^{z}_{Y(0)}, \mu^{z}_{Y(0)}) \diff \mmu_{Z(0)}(z)\right]
        \leq & C'_4 \eta^{-1} R_{d_Y}(n^{1-rd_Z}) + L_Z \sqrt{d_Z}  n^{-r}.
    \end{align*}
    $R_{d_Y}(\cdot)$ is defined in \eqref{eq:R}, and $C'_4$ is a constant that depends on $d_Y$ and $\max_{y \in \calY} \lnorm{y}{2}$.
    Therefore, we get
    \begin{align*}
        \EE[(\textup{Term A})] \leq C'_5(R_{d_Y}((m \wedge n)^{1-rd_Z}) + (m \wedge n)^{-r}).
    \end{align*}
    Let $r = \frac{1}{d_Z + 2 \vee d_Y}$, the remaining steps follow the same as the proof of Theorem~\ref{thm:finite_sample}.
\end{proof}

\section{Proof of Section \ref{sec:simulation}}

\begin{lemma}[$V_{\cc}$ for Gaussian noise model, {\cite[Lemma 6.1]{lin2025tightening}}]\label{lem:gauss_noise_Vc}
    Assume that $h(y, \bar y) = \lnorm{y - \bar y}{2}^2$, and the noise variables follow the distribution $\varepsilon_k \sim N(0, \Sigma_k), k=0,1$. Then,
    \begin{enumerate}[label=(\roman*), leftmargin=0.3in]
        \item For the location model ($m = 1$), $V_{\cc} =\EE_{Z}\left[\lnorm{f_0(Z) - f_1(Z)}{2}^2\right] + S(\Sigma_0, \Sigma_1)$.

        \item For the scale model ($m = 2$), $V_{\cc} = \EE_{Z}\left[S(\calD(f_0(Z))\Sigma_0\calD(f_0(Z)), \calD(f_1(Z))\Sigma_1\calD(f_1(Z)))\right].$
    \end{enumerate}
    Here, $S(\Sigma_0, \Sigma_1) = \text{Tr}\left(\Sigma_0 + \Sigma_1 - 2(\Sigma_0^{\frac{1}{2}} \Sigma_1 \Sigma_0^{\frac{1}{2}})^{\frac{1}{2}}\right)$, and $\calD(v)$ is the matrix with $v$ at the diagonal.
\end{lemma}

\section{Proofs of Supporting Lemmas}\label{sec:proof_lemma}
\subsection{Proof of Lemma~\ref{lem:tri_ineq}}
\begin{proof}[Proof of Lemma~\ref{lem:tri_ineq}]
    Let $\pi$ ($\tilde \pi$, resp.) be an optimal coupling associated with $\Wass_1(\mu, \nu)$ ($\Wass_1(\nu, \mu')$, resp.). By the gluing lemma in Chapter 1 of \cite{villani2009optimal}, there is a coupling $\gamma = \gamma(x,y,x') $ such that
    \[
        \gamma_{xy} = \pi,\quad \gamma_{yx'} = \tilde \pi.
    \]

    As a result, we have
    \[
        \begin{aligned}
            \Wass_1(\mu, \mu') 
            \leq & \int \lnorm{x - x'}{2} \diff \gamma(x,y,x') \\
            \leq & \int (\lnorm{x - y}{2} + \lnorm{y - x'}{2}) \diff \gamma(x,y,x')\\
            = & \int \lnorm{x - y}{2} \diff \pi(x,y) + \int \lnorm{y - x'}{2} \diff \tilde \pi(y, x')\\
            = & \Wass_1(\mu, \nu) + \Wass_1(\nu, \mu').
        \end{aligned}
    \]
\end{proof}

\subsection{Proof of Lemma~\ref{lem:convex}}
\begin{proof}[Proof of Lemma~\ref{lem:convex}]
    Let $\pi$ ($\tilde \pi$, resp.) be an optimal coupling associated with $\Wass_1(\mu, \mu')$ ($\Wass_1(\nu, \mu')$, resp.). Then, for $\lambda \in [0,1]$, $\lambda\pi + (1-\lambda)\tilde \pi$ is a coupling for $\lambda\mu + (1-\lambda)\nu$ and $\mu'$. 

    As a result, we have
    \[
        \begin{aligned}
            \Wass_1(\lambda\mu + (1-\lambda)\nu, \mu') 
            &\leq \int \lnorm{x - x'}{2} \diff (\lambda\pi + (1-\lambda)\tilde \pi)\\
            &= \lambda  \int \lnorm{x - x'}{2} \diff \pi + (1-\lambda) \int \lnorm{x - x'}{2} \diff \tilde \pi\\
            &= \lambda \Wass_1(\mu, \mu') + (1-\lambda) \Wass_1(\nu, \mu').
        \end{aligned}
    \]
\end{proof}

\subsection{Proof of Lemma~\ref{lem:Pic_compact}}
\begin{proof}[Proof of Lemma~\ref{lem:Pic_compact}]
    Since $\Pi_{\cc} \subseteq \Pi(\mu_{Y(0)}, \mu_{Y(1)})$ and $\Pi(\mu_{Y(0)}, \mu_{Y(1)})$ is tight (\cite[Lemma 4.4]{villani2009optimal}), then $\Pi_{\cc}$ is also tight. It remains to show that $\Pi_{\cc}$ is closed with respect to the weak topology. Suppose that $\pi_{k} \in \Pi_{\cc}$ for $k \geq 1$ and $\lim_{k} \pi_{k} = \pi_{\infty}$, then for any bounded continuous function $\ell: \calY \times \calZ \rightarrow \RR$, $\EE_{\pi_{\infty}}[\ell(Y(0), Z)] = \lim_{k} \EE_{\pi_{k}}[\ell(Y(0), Z)] = \int \ell \diff \mu_{Y(0), Z}$. Similarly, $\EE_{\pi_{\infty}}[\ell(Y(1), Z)] = \int \ell \diff \mu_{Y(1), Z}$. Therefore, $\pi_{\infty} \in \Pi_{\cc}$, thus $\Pi_{\cc}$ is closed.
\end{proof}

\subsection{Proof of Lemma~\ref{lemm:unicont_kernel}}
\begin{proof}[Proof of Lemma~\ref{lemm:unicont_kernel}]
    On the compact space $\calY$, the weak topology can be induced by the metric $\Wass_1$, thus $g_w(z) = \mu_{Y(w)}^z\ w=0,1$ are continuous with respect to $\Wass_1$. Since $\calZ$ is compact, then $g_w(z)\ w=0,1$ are uniformly continuous with respect to $\Wass_1$. As a result, we get that, for any $\delta > 0$, there is a constant $C(\delta)$ such that $\Wass_1\left(\mu_{Y(k)}^{z_0}, \mu_{Y(k)}^{z_1}\right) \leq \delta + C(\delta)\|z_0 - z_1\|_2$.
\end{proof}

\subsection{Proof of Lemma~\ref{lemm:LipWassh}}
\begin{proof}[Proof of Lemma~\ref{lemm:LipWassh}]
    We only need to prove that 
    \[
        \left| \Wass_{h}(\mu, \nu) - \Wass_{h}(\mu, \nu') \right| \leq L_h  \Wass_1(\nu, \nu') + \epsilon.
    \]
    If this holds, then the desired result is followed by applying the triangle inequality:
    \begin{equation}\label{eq:split_tri}
        \left| \Wass_{h}(\mu, \nu) - \Wass_{h}(\mu', \nu') \right| \leq 
        \left| \Wass_{h}(\mu, \nu) - \Wass_{h}(\mu, \nu') \right| + \left| \Wass_{h}(\mu, \nu') - \Wass_{h}(\mu', \nu') \right| + 2\epsilon.
    \end{equation}

    In the following, we use the similar technique for proving the triangle inequality of Wasserstein distances.
        
    Let $\pi_{1} \in \calP(\calY^2)$ be the optimal coupling associated with $\Wass_h(\mu, \nu)$, and let $\pi_2 \in \calP(\calY^2)$ be the optimal coupling associated with $\Wass_1\left(\nu, \nu'\right)$.
    By the gluing lemma in Chapter 1 of \cite{villani2009optimal}, there is a coupling $\pi_3 \in \calP(\calY^3)$, where we denote $\calY^3 = \{(y^{(1)}, y^{(2)}, y^{(3)}), y^{(k)} \in \calY\}$, such that  
    \begin{align*}
        &\pi_3(y^{(1)}, y^{(2)}) = \pi_1,\\
        &\pi_3(y^{(2)}, y^{(3)}) = \pi_2.
    \end{align*}
    
    Therefore, by Assp.~\ref{a:h_func}, we have
    \begin{align*}
        \Wass_h(\mu, \nu') \leq & \EE_{\pi_3}\left[h(y^{(1)}, y^{(3)})\right]\\
        \leq & \EE_{\pi_3}\left[h(y^{(1)}, y^{(2)}) + L_h\|y^{(2)} - y^{(3)}\|_2 + \epsilon \right]\\
        = & \Wass_h(\mu, \nu) + L_h \Wass_1\left(\nu, \nu'\right) + \epsilon.
    \end{align*}
    where the first inequality is due to $\pi_3(y^{(1)}, y^{(3)}) \in \Pi(z, z_1)$; the second inequality is due to the smoothness assumption of $h$ \eqref{eq:hsmoothassp}; the equation is due to the construction of $\pi_3$ by gluing $\pi_1, \pi_2$. 

    Swapping $\nu$ and $\nu'$, we can the other direction of the inequality. Jointly, we get
    \[
        \left|\Wass_h(\mu, \nu) - \Wass_h(\mu, \nu') \right| \leq L_h \Wass_1\left(\nu, \nu'\right) + \epsilon.
    \]
    Plug it back to \eqref{eq:split_tri}, we get the desired result.
\end{proof}

\subsection{Proof of Lemma~\ref{lemm:var_wassh}}
\begin{proof}[Proof of Lemma~\ref{lemm:var_wassh}]
    Let $\pi_z$ be the optimal coupling associated with $\Wass_h(\mu_{Y(0)}^z, \mu_{Y(1)}^z)$, then we have, for a point $z_0 \in \calZ$,
    \begin{align*}
        &\int \Wass_h(\mu_{Y(0)}^{z}, \mu_{Y(1)}^{z})^2 \diff \mu_Z(z)\\
        = & \int (\Wass_h(\mu_{Y(0)}^{z}, \mu_{Y(1)}^{z}) - \Wass_h(\mu_{Y(0)}^{z_0}, \mu_{Y(1)}^{z_0}) + \Wass_h(\mu_{Y(0)}^{z_0}, \mu_{Y(1)}^{z_0}))^2 \diff \mu_Z(z)\\
        \leq & 2 \int (\Wass_h(\mu_{Y(0)}^{z}, \mu_{Y(1)}^{z}) - \Wass_h(\mu_{Y(0)}^{z_0}, \mu_{Y(1)}^{z_0}))^2 + (\Wass_h(\mu_{Y(0)}^{z_0}, \mu_{Y(1)}^{z_0}))^2 \diff \mu_Z(z)\\
        = & \underbrace{2 \int (\Wass_h(\mu_{Y(0)}^{z}, \mu_{Y(1)}^{z}) - \Wass_h(\mu_{Y(0)}^{z_0}, \mu_{Y(1)}^{z_0}))^2 \diff \mu_Z(z)}_{ \text{(Term A)} } + C.
    \end{align*}

    For \textbf{(Term A)}, by Lemma~\ref{lemm:LipWassh} and Assp.~\ref{a:Lip_kernel}, we have
    \[
        (\Wass_h(\mu_{Y(0)}^{z}, \mu_{Y(1)}^{z}) - \Wass_h(\mu_{Y(0)}^{z_0}, \mu_{Y(1)}^{z_0}))^2 \leq 4 (L_hL_Z)^2 \|z - z_0\|_2^2 \leq 8 (L_hL_Z)^2 (\|z\|_2^2 + \|z_0\|_2^2).
    \]
    Take it back, we get
    \[
        \text{(Term A)} \leq 16 (L_hL_Z)^2 \int \|z\|_2^2 \diff \mu_Z(z) + \tilde C,
    \]
    which proves the desired result.
\end{proof}

\subsection{Proof of Lemma~\ref{lemm:gap_mmu}}
\begin{proof}[Proof of Lemma~\ref{lemm:gap_mmu}]
        Note that $\mmu_{Z(0)}$ is the push-forward measure of $\hat \mu_{Z(0)}$ under $\varphi^n_r$. Then,
        \begin{align*}
            \Wass_{1}(\hat \mu_{Z(0)}, \mmu_{Z(0)}) \leq & \left(\int \lnorm{\zeta - \varphi_r^n(\zeta)}{2} \diff \hat \mu_{Z(0)}(\zeta)\right)\\
            \leq & \sup_{\zeta \in [0,1]^{d_Z}} \lnorm{\zeta - \varphi_r^n(\zeta)}{2} & (\textup{due to Assp.~\ref{a:compact_domain}})\\
            \leq & \sqrt{d_Z}  n^{-r}.
        \end{align*}
\end{proof}

\subsection{Proof of Lemma~\ref{lemm:gap_mmu2}}
\begin{lemma}\label{lemm:C7}
    For any $G \in \Phi^n_r$, we have: (i) $\hat \mu_{Y(0)}^{G} = \mmu_{Y(0)}^{G}$, and (ii) $\mmu_{Z(0)}(G) = \hat \mu_{Z(0)}(G)$.
\end{lemma}
\begin{proof}[Proof of Lemma~\ref{lemm:C7}]
    Note that, due to the construction of $\varphi_r^n$ and $G \in \Phi_r^n$, we have $I_G = \tilde I_G$. As a result, we get $\hat \mu_{Y(0)}^{G} = \mmu_{Y(0)}^{G}$, which proves (i). For (ii), note that $n \hat \mu_{Z(0)}(G) = |I_G| = |\tilde I_G| = n \mmu_{Z(0)}(G)$.
\end{proof}
\begin{proof}[Proof of Lemma~\ref{lemm:gap_mmu2}]
     We have
        \begin{align}
            \int \Wass_{1}(\mmu_{Y(0)}^{z}, \mu_{Y(0)}^{z}) \diff \mmu_{Z(0)}(z)
            =& \sum_{\zeta\in \varphi_r^n([0,1]^{d_Z})} \Wass_{1}(\mmu_{Y(0)}^{\zeta}, \mu_{Y(0)}^{\zeta}) \mmu_{Z(0)}(\zeta)\label{eq:eq26}.
        \end{align}
        
        Note that, in \eqref{eq:eq26}, we have $\mmu_{Y(0)}^{\zeta} = \mmu_{Y(0)}^{G}$, and $\mmu_{Z(0)}(\zeta) = \mmu_{Z(0)}(G)$ for $G = (\varphi_r^n)^{-1}(\zeta) \in \Phi_r^n$.

        Then, by Lemma~\ref{lemm:C7}, we have
        \begin{align*}
            \Wass_{1}(\mmu_{Y(0)}^{G}, \mu_{Y(0)}^{\zeta}) \leq & \Wass_{1}(\mmu_{Y(0)}^{G}, \mu_{Y(0)}^{G}) + \Wass_{1}(\mu_{Y(0)}^{G}, \mu_{Y(0)}^{\zeta}) & (\textup{by the triangle ineq.})\\
            \leq & \Wass_{1}(\hat \mu_{Y(0)}^{G}, \mu_{Y(0)}^{G}) + \Wass_{1}(\mu_{Y(0)}^{G}, \mu_{Y(0)}^{\zeta}). & (\textup{by Lemma~\ref{lemm:C7}})
        \end{align*}

        By the convexity of $\mu \mapsto \Wass_1(\mu, \mu_{Y(0)}^{\zeta})$ (Lemma~\ref{lem:convex}) and Assp.~\ref{a:Lip_kernel}, for $\zeta \in \varphi_r^n([0,1]^{d_Z})$, we have
        \begin{align*}
            \Wass_{1}(\mu^G_{Y(0)}, \mu^{\zeta}_{Y(0)}) = & \Wass_{1}\left(\int_{G} \mu^{\zeta'}_{Y(0)} \diff \mu_{Z}(\zeta'), \mu^{\zeta}_{Y(0)}\right)\\
            \leq & \int_G \Wass_1(\mu^{\zeta'}_{Y(0)}, \mu^{\zeta}_{Y(0)}) \diff \mu_{Z}(\zeta')\\
            \leq & \int_G L_{Z} \lnorm{\zeta' - \zeta}{2} \diff \mu_{Z}(\zeta') & (\textup{due to Assp.~\ref{a:Lip_kernel}})\\
            \leq & L_{Z} \sup_{\zeta' \in G}\lnorm{\zeta - \zeta'}{2}. & (\textup{due to Assp.~\ref{a:compact_domain}})\\
            \leq & L_{Z} \sqrt{d_Z}  n^{-r}
        \end{align*}

        Back to \eqref{eq:eq26}, by Lemma~\ref{lemm:C7} that $\hat \mu_{Z(0)}(G) = \mmu_{Z(0)}(G)$, we obtain
        \begin{align*}
            & \sum_{\zeta\in \varphi_r^n([0,1]^{d_Z})} \Wass_{1}(\mmu_{Y(0)}^{\zeta}, \mu_{Y(0)}^{\zeta}) \mmu_{Z(0)}(\zeta)\\
            \leq & \sum_{G \in \Phi_r^n} \mmu_{Z(0)}(G) \left(\Wass_{1}(\hat \mu_{Y(0)}^{G}, \mu_{Y(0)}^{G}) + L_Z \sqrt{d_Z}  n^{-r}\right)\\
            = & \sum_{G \in \Phi_r^n} \mmu_{Z(0)}(G) \Wass_{1}(\hat \mu_{Y(0)}^{G}, \mu_{Y(0)}^{G}) + L_Z \sqrt{d_Z}  n^{-r}\\
            = & \sum_{G \in \Phi_r^n} \hat \mu_{Z(0)}(G) \Wass_{1}(\hat \mu_{Y(0)}^{G}, \mu_{Y(0)}^{G}) + L_Z \sqrt{d_Z}  n^{-r}. & (\textup{by Lemma~\ref{lemm:C7}})
        \end{align*}
        Plugging back to \eqref{eq:eq26} proves the desired result.
\end{proof}

\subsection{Proof of Lemma~\ref{lem:TV_bound}}
\begin{proof}[Proof of Lemma~\ref{lem:TV_bound}]
    Note that $d_{\text{TV}}(Z,Z') = \inf_{\lambda \in \Pi(\QQ, \PP)}~\EE_{\lambda}\left[\mathbf{1}(Z \neq Z')\right]$. Then, since 
    \[
        \lnorm{Z - Z'}{2} \leq r \times \mathbf{1}(Z \neq Z'),
    \]
    then we get 
    \[
        \Wass_{1}(Z,Z') \leq r \times d_{\text{TV}}(Z,Z').
    \]
\end{proof}

\subsection{Proof of Lemma~\ref{lem:linear_gap}}
\begin{proof}[Proof of Lemma~\ref{lem:linear_gap}]
    By Lemma~\ref{lem:TV_bound}, 
    \[
        \Wass_1(\hat \mu_{Z(0), n}, \tilde \mu_{Z(0), n}) \leq \sqrt{d_{Z}} d_{\text{TV}}(\hat \mu_{Z(0), n}, \tilde \mu_{Z(0), n}).
    \]

    Then, we have
    \[
        \begin{aligned}
        &\Wass_1(\hat \mu_{Z(0), n}, \tilde \mu_{Z(0), n})\\
        \leq &
        \frac{\sqrt{d_Z}}{2} \sum_{i=1}^n \left|\hat w_i - \frac{\frac{\diff \mu_Z}{\diff \mu_{Z|W = 0}}(Z_i(0))}{\sum_{i=1}^{n} \frac{\diff \mu_Z}{\diff \mu_{Z|W = 0}}(Z_i(0))}\right|\\
        \leq & \frac{\sqrt{d_Z}}{2} \sum_{i=1}^n \left|\frac{\sum_{i=1}^n\frac{\diff \mu_Z}{\diff \mu_{Z|W = 0}}(Z_i(0))\hat w_i - \frac{\diff \mu_Z}{\diff \mu_{Z|W = 0}}(Z_i(0))}{\sum_{i=1}^{n} \frac{\diff \mu_Z}{\diff \mu_{Z|W = 0}}(Z_i(0))}\right| \\
        \leq & \frac{\sqrt{d_Z}}{2} \sum_{i=1}^n \left|\frac{\sum_{i=1}^n\frac{\diff \mu_Z}{\diff \mu_{Z|W = 0}}(Z_i(0)) - n}{\sum_{i=1}^{n} \frac{\diff \mu_Z}{\diff \mu_{Z|W = 0}}(Z_i(0))}\right| \hat w_i + \frac{\sqrt{d_Z}}{2}\sum_{i=1}^n \left|\frac{n\hat w_i - \frac{\diff \mu_Z}{\diff \mu_{Z|W = 0}}(Z_i(0))}{\sum_{i=1}^{n} 
        \frac{\diff \mu_Z}{\diff \mu_{Z|W = 0}}(Z_i(0))}\right| \\
        \leq & \underbrace{\frac{\sqrt{d_Z}}{2} \left|\frac{\sum_{i=1}^n\frac{\diff \mu_Z}{\diff \mu_{Z|W = 0}}(Z_i(0)) - n}{\sum_{i=1}^{n} \frac{\diff \mu_Z}{\diff \mu_{Z|W = 0}}(Z_i(0))}\right|}_{(\text{Term A})} + \underbrace{\frac{\sqrt{d_Z}}{2\delta n}\sum_{i=1}^n \left|n\hat w_i - \frac{\diff \mu_Z}{\diff \mu_{Z|W = 0}}(Z_i(0))\right|}_{(\text{Term B})}.
        \end{aligned}
    \]
    Here, for the last inequality, we apply $\sum_{i=1}^n w_i = 1$ and $ \frac{\diff \mu_{Z}}{\diff \mu_{Z|W=0}}(z) \geq \PP(W = 0)~\forall z$, where $\PP(W = 0) = \int \PP(W=0 | Z=z) \diff \mu_Z(z) \geq \delta$ by Assp.~\ref{assu:overlap}.

    For \textbf{(Term A)}, since for $x,y \geq \delta$, $x^{-1} - y^{-1} \leq \delta^{-2}|x - y|$, then we have
    \[
        \begin{aligned}
            \EE[(\text{Term A})] &= \frac{\sqrt{d_Z}}{2} \EE\left|1 - \frac{1}{\frac{1}{n}\sum_{i=1}^{n} \frac{\diff \mu_Z}{\diff \mu_{Z|W = 0}}(Z_i(0))}\right|\\
            &\leq \frac{\sqrt{d_Z}}{2\delta^2} \EE\left|\frac{1}{n}\sum_{i=1}^{n} \frac{\diff \mu_Z}{\diff \mu_{Z|W = 0}}(Z_i(0)) - 1\right|\\
            &\leq \frac{\sqrt{d_Z}}{2\delta^2\sqrt{n}} \sqrt{\text{Var}\left(\frac{\diff \mu_Z}{\diff \mu_{Z|W = 0}}(Z)\right)}\\
            &\leq \frac{\sqrt{d_Z}}{2\delta^3\sqrt{n}}.
        \end{aligned}
    \]
    Here, for the last inequality, we apply $\frac{\diff \mu_Z}{\diff \mu_{Z|W = 0}}(z) = \PP(W = 0)/(1-e(z)) \leq \delta^{-1}$.

    For \textbf{(Term B)}, we apply Assp.~\ref{a:est_propensity_score} and get
    \[
        \EE[(\text{Term B})] \leq \frac{\sqrt{d_Z} C_{\textup{w}}}{2 \delta} n^{-r}.
    \]

    Combine (Term A) and (Term B), we get the desired result.
\end{proof}

\subsection{Proof of Lemma~\ref{lem:Wass_conv_reweight}}
\begin{proof}[Proof of Lemma~\ref{lem:Wass_conv_reweight}]
    We write $Z_i(0)$ as $Z_i$ in this proof.
   By the triangle inequality (Lemma~\ref{lem:tri_ineq}),
    \[
        \EE[\Wass_1(\hat \mu_{Z(0), n}, \mu_{Z(0)})] \leq \underbrace{\EE[\Wass_1(\hat \mu_{Z(0), n}, \tilde \mu_{Z(0), n})]}_{(\text{Term A})} + \underbrace{\EE[\Wass_1(\tilde \mu_{Z(0), n}, \mu_{Z(0)})]}_{(\text{Term B})}.
    \]

    For \textbf{(Term A)}, apply Lemma~\ref{lem:linear_gap}, we get 
    \[
        \text{(Term A)} \leq \frac{\sqrt{d_Z}}{2\delta^3}n^{-\frac{1}{2}} + \frac{\sqrt{d_Z} C_{\textup{w}}}{2 \delta} n^{-r}.
    \]

    For \textbf{(Term B)}, we follow the similar argument of the proof of \cite[Theorem 1]{fournier2015rate}, and see that it depends on the bound of $\EE[|\tilde \mu_{Z(0), n}(A) - \mu_{Z(0)}(A)|]$ for any subset $A$ of $\RR^{d_Z}$.

    Indeed, we have,
    \begin{align*}
        &\EE[|\tilde \mu_{Z(0), n}(A) - \mu_{Z(0)}(A)|] \\
        =& \EE\left[\left|\sum_{i=1}^n w_i \delta_{Z_i}(A) - \mu_{Z(0)}(A)\right|\right]\\
        =& \EE\left[\left|\frac{\sum_{i=1}^{n} \frac{\diff \mu_{Z}}{\diff \mu_{Z|W=0}}(Z_i) (\delta_{Z_i}(A) - \mu_{Z(0)}(A)) }{\sum_{i=1}^{n} \frac{\diff \mu_{Z}}{\diff \mu_{Z|W=0}}(Z_i)}\right|\right]\\
        \leq & \frac{1}{\PP(W=0)} \EE\left[\left|\frac{1}{n} \sum_{i=1}^{n} \frac{\diff \mu_{Z}}{\diff \mu_{Z|W=0}}(Z_i) (\delta_{Z_i}(A) - \mu_{Z(0)}(A)) \right|\right]\\
        \leq & \frac{1}{\sqrt{n} \PP(W=0)} \left(\EE_{Z \sim \mu}\left[\left(\frac{\diff \mu_{Z}}{\diff \mu_{Z|W=0}}(Z)\delta_{Z}(A) - \mu_{Z(0)}(A)\right)^2\right]\right)^{\frac{1}{2}} & (\text{Jensen's inequality})\\
        \leq & \sqrt{\frac{\mu_{Z(0)}(A)}{n}} \cdot \frac{\delta^{-1}}{\PP(W=0)}.
    \end{align*}
    In addition, we have
    \begin{align*}
        \EE\left[|\tilde \mu_{Z(0)}(A) - \mu_{Z(0)}(A)|\right] \leq \EE\left[\tilde \mu_{Z(0)}(A)\right] + \mu_{Z(0)}(A) \leq \frac{2}{\PP(W = 0)} \cdot \mu_{Z(0)}(A),
    \end{align*}
    where we apply the inequality $\PP(W = 0) \leq \frac{\diff \mu_{Z}}{\diff \mu_{Z|W=0}}(Z_i) \leq \delta^{-1}$.

    Further, $\PP(W = 0) = \int \PP(W=0 | Z=z) \diff \mu_Z(z) \geq \delta$ by Assp.~\ref{assu:overlap}.
    
    Jointly, we have
    \begin{align*}
        \EE[|\tilde \mu_{Z(0)}(A) - \mu_{Z(0)}(A)|] \leq \delta^{-2}\min\left(2 \mu_{Z(0)}(A), \sqrt{\frac{\mu_{Z(0)}(A)}{n}}\right).
    \end{align*}
    This inequality is similar to the first inequality in \cite[Section 3]{fournier2015rate}, except that instead of $\delta^{-2}$, they have $1$ as the constant.

    Then, following the same argument as in the proof of \cite[Theorem 1]{fournier2015rate}, we get
    \[
        \text{(Term B)} \leq \EE[\Wass_1(\tilde \nu, \nu)] \leq C\delta^{-2} R_{d_Z}(n),
    \]
    where $C$ is a constant depending on $d_Z$.

    Combining (Term A) and (Term B), we get the desired result.    
\end{proof}

\subsection{Proof of Lemma~\ref{lem:discrete_gapII}}
\begin{proof}[Proof of Lemma~\ref{lem:discrete_gapII}]
    Let 
    \[
        W_G = \frac{(1-\hat e(\varphi_r^n(G)))^{-1} |\{i: Z_i(0) \in G\}|}{\sum_{G} (1-\hat e(\varphi_r^n(G)))^{-1} |\{i: Z_i(0) \in G\}|},
    \]
    then by the triangle inequality (Lemma~\ref{lem:tri_ineq}), we have
    \[
        \begin{aligned}
            &\Wass_1(\mmu_{Z(0), n}, \hat \mu_{Z(0), n})\\
            = & \Wass_1(\sum_{G \in \Phi_r^n} W_G\delta_{\varphi_r^n(Z_i(0))}, \sum_{i=1}^{n} \hat w_i \delta_{Z_i(0)})\\
            \leq & \Wass_1(\sum_{G \in \Phi_r^n} W_G\delta_{\varphi_r^n(Z_i(0))}, \sum_{G \in \Phi_r^n} \frac{W_G}{|\{i: Z_i(0) \in G\}|} \sum_{i\in G} \delta_{Z_i(0)})\\
            & + \Wass_1(\sum_{G \in \Phi_r^n} \frac{W_G}{|\{i: Z_i(0) \in G\}|}\sum_{i\in G} \delta_{Z_i(0)}, \sum_{i=1}^{n} \hat w_i \delta_{Z_i(0)})\\
            \leq & \sqrt{d_Z} n^{-r} + \underbrace{\Wass_1(\sum_{G \in \Phi_r^n} \frac{W_G}{|\{i: Z_i(0) \in G\}|}\sum_{i\in G}  \delta_{Z_i(0)}, \sum_{i=1}^{n} \hat w_i \delta_{Z_i(0)})}_{(\text{Term A})},
        \end{aligned}
    \]
    where the last inequality is due to Lemma~\ref{lemm:gap_mmu}.

    For \textbf{(Term A)}, by Lemma~\ref{lem:TV_bound}, we have
    \[
        \begin{aligned}
            (\text{Term A})
            \leq & \frac{\sqrt{d_Z}}{2} \sum_{G \in \Phi_r^n} \sum_{i\in G} \left| \frac{W_G}{|\{i: Z_i(0) \in G\}|} - \hat w_i\right|\\
            \leq & \frac{\sqrt{d_Z}}{2} \sum_{G \in \Phi_r^n} \sum_{i\in G} \left(\underbrace{\frac{(1-\hat e(\varphi_r^n(G)))^{-1} - (1 - \hat e(Z_i(0)))^{-1}}{\sum_{G} (1-\hat e(\varphi_r^n(G)))^{-1} |\{i: Z_i(0) \in G\}|}}_{(\text{Term B})_i} \right.\\
            & \left. + \underbrace{\frac{(1 - \hat e(Z_i(0)))^{-1} }{\sum_{G} (1-\hat e(\varphi_r^n(G)))^{-1} |\{i: Z_i(0) \in G\}|} -  \frac{(1 - \hat e(Z_i(0)))^{-1} }{\sum_{i=1}^n (1-\hat e(Z_i(0)))^{-1}}}_{(\text{Term C})_i}\right).
        \end{aligned}
    \]
    
    For \textbf{$(\text{Term B})_i$}, by Assp.~\ref{a:lip_propensity_score} and~\ref{a:bound_e}, we have
    \[
        \begin{aligned}
            |(1-\hat e(\varphi_r^n(G)))^{-1} - (1 - \hat e(Z_i(0)))^{-1}|
            \leq \frac{L_e}{\eta^2} n^{-r},\quad
            \sum_{G} (1-\hat e(\varphi_r^n(G)))^{-1} |\{i: Z_i(0) \in G\}| \geq n.
        \end{aligned}
    \]
    Thus,
    \[
        \sum_{G \in \Phi_r^n} \sum_{i\in G} |\text{(Term B)}_i| \leq \frac{L_e}{\eta^2} n^{-r}.
    \]
    For \textbf{$(\text{Term C})_i$}, by the same reasoning above, we have
    \[
        \begin{aligned}
            |(\text{Term C})_i|
            &= \frac{|\sum_{i=1}^n (1-\hat e(Z_i(0)))^{-1} - \sum_{G} (1-\hat e(\varphi_r^n(G)))^{-1} |\{i: Z_i(0) \in G\}||}{(\sum_{G} (1-\hat e(\varphi_r^n(G)))^{-1} |\{i: Z_i(0) \in G\}|)(\sum_{i=1}^n (1-\hat e(Z_i(0)))^{-1})}\\
            &\leq \frac{n \frac{L_e}{\eta^2} n^{-r}}{(\sum_{G} (1-\hat e(\varphi_r^n(G)))^{-1} |\{i: Z_i(0) \in G\}|)(\sum_{i=1}^n (1-\hat e(Z_i(0)))^{-1})} \leq \frac{1}{n} \frac{L_e}{\eta^3} n^{-r}.
        \end{aligned}
    \]
    Thus, we get
    \[
        \sum_{G \in \Phi_r^n} \sum_{i\in G} |\text{(Term C)}_i|
        \leq \frac{L_e}{\eta^3} n^{-r}.
    \]
    Combine these, we get
    \[
        \Wass_1(\mmu_{Z(0), n}, \hat \mu_{Z(0), n}) \leq \sqrt{d_Z} (1+\frac{L_e}{\eta^3}) n^{-r}.
    \]
    
\end{proof}